\documentclass{AIMS}
\usepackage{amsmath,amssymb,enumerate,amsthm}
  \usepackage{paralist}
  \usepackage{graphics} 
  \usepackage{epsfig} 
 \usepackage[colorlinks=true]{hyperref}
\usepackage{pb-diagram}
\usepackage{tikz}
\usetikzlibrary{matrix,arrows}

\hypersetup{urlcolor=blue, citecolor=red}

\def\currentvolume{X}
 \def\currentissue{X}
  \def\currentyear{200X}
   \def\currentmonth{XX}
    \def\ppages{X--XX}

\newtheorem{theorem}{Theorem}[section]
\newtheorem{corollary}{Corollary}[section]

\newtheorem{lemma}[theorem]{Lemma}
\newtheorem{proposition}{Proposition}[section]

\theoremstyle{definition}

\newtheorem{remark}{Remark}[section]
\newtheorem{example}{Example}[section]

\DeclareMathOperator*{\Aut}{Aut}

\numberwithin{equation}{section}

\numberwithin{figure}{section}

 \numberwithin{table}{section}

\title[Skew Constacyclic Codes over Finite Chain Rings]
      {Skew Constacyclic Codes over Finite Chain Rings}

\author[Somphong Jitman, San Ling and Patanee
Udomkavanich]{}


 \email{pu738241@e.ntu.edu.sg}
 \email{lingsan@ntu.edu.sg}
 \email{pattanee.u@chula.ac.th}

\thanks{
The research of the first and second authors is partially supported by the Singapore Ministry of Education under
Research Grant T208B2206.  The first author is also supported by the Institute for the Promotion of Teaching
Science and Technology of Thailand. }

\begin{document}
\maketitle

\centerline{\scshape Somphong Jitman}
\medskip
{\footnotesize
\centerline{Department of Mathematics, Faculty of Science,}
\centerline{Chulalongkorn University, Bangkok  10330, Thailand} %
\centerline{and} %
\centerline{Division of Mathematical Sciences,}%
\centerline{School of Physical and Mathematical Sciences, Nanyang Technological University,}%
\centerline{21 Nanyang Link, Singapore 637371, Republic of Singapore}
} 

\medskip

\centerline{\scshape San Ling}
\medskip
{\footnotesize
 \centerline{Division of Mathematical Sciences,}%
\centerline{School of Physical and Mathematical Sciences, Nanyang Technological University,}%
\centerline{21 Nanyang Link, Singapore 637371, Republic of Singapore} }

\medskip

\centerline{\scshape Patanee Udomkavanich}
\medskip
{\footnotesize
\centerline{Department of Mathematics, Faculty of Science,}%
\centerline{Chulalongkorn University, Bangkok 10330, Thailand} }
\bigskip

 \centerline{(Communicated by the associate editor name)}

\begin{abstract}
Skew polynomial rings over finite fields (\cite{BoGe2007} and \cite{BoUl2009}) and over Galois rings
(\cite{BoSoUm2008}) have been used to study codes. In this paper, we extend this concept to finite chain rings.
Properties of skew constacyclic codes generated by monic right divisors of $x^n-\lambda$, where $\lambda$ is a
unit element, are exhibited. When $\lambda^2=1$, the generators of Euclidean and Hermitian dual  codes of such
codes are determined together with necessary and sufficient conditions for them to be Euclidean and Hermitian
self-dual. Of more interest are codes over the ring $\mathbb{F}_{p^m}+u\mathbb{F}_{p^m}$. The structure of all
skew constacyclic codes
 is completely determined. This allows us to express generators of
Euclidean and Hermitian dual codes of skew cyclic and skew negacyclic codes   in terms of the generators of the
original codes. An illustration of all skew cyclic codes of length~$2$ over $\mathbb{F}_{3}+u\mathbb{F}_{3}$ and
their Euclidean and Hermitian dual codes is also provided.
\end{abstract}

\section{Introduction}

In the early history of the art of Error Correcting Codes, codes were usually taken over finite fields. In the
last two decades,  interest has been shown in linear codes over rings.   In an important work \cite{HaKu1994},
Calderbank et al. showed that the Kerdock codes, the  Preparata codes and Delsarte-Goethals codes can be
obtained through the Gray images of linear codes over $\mathbb{Z}_4$. In general,  due to their rich algebraic
structure, constacyclic codes have been studied  over various finite chain rings (see, for example,
\cite{AmNe2008}, \cite{BoUd1999}, \cite{HQDSRLP2004}, \cite{Di2005}, \cite{Di2009}, \cite{Di2010},
\cite{QiZa2006}, \cite{NoSa2000},
   \cite{SoEs2009-arbi} and \cite{UdBo1999}). In particular, successful applications of modular lattices using  codes over a finite chain ring $\mathbb{F}_{p}+u\mathbb{F}_{p}$~\cite{Ba1997} and
constructions of good sequences from polynomial residue class rings \cite{Udsi1998} have motivated the study of
constacyclic  codes over a special family of finite chain rings of the form $\mathbb{F}_{p^m}+u\mathbb{F}_{p^m}$
(see, for example, \cite{AmNe2008}, \cite{BoUd1999}, \cite{Di2009}, \cite{Di2010}, \cite{QiZa2006} and
\cite{UdBo1999}).

 Classically, polynomial rings over finite fields or over finite rings and their ideals are
key  to determining the algebraic structures of constacyclic  codes. In~\cite{BoGe2007}, skew (non-commutative)
polynomial rings have been used to
 describe the structure of linear codes closed under  a skew
cyclic shift, namely, skew cyclic codes. Later on, in \cite{BoUl2009}, more properties  and good examples of
such  codes have been established. Recently, in \cite{BoSoUm2008}, that approach has been extended to codes over
Galois rings. Skew constacyclic codes have been studied for a particular case when codes are generated by monic
right divisors of $x^n-\lambda$, where $\lambda$ is a unit in the Galois ring fixed by a given automorphism.

Motivated by these works, we generalize the concept of skew constacyclic codes to over finite chain rings. We
study the algebraic structure and properties of these codes and their Euclidean and Hermitian dual codes. For
arbitrary finite chain rings,  skew constacyclic codes with respect to a unit $\lambda$ are studied in the case
where their generator polynomials are right divisors of $x^n-\lambda$.  Moreover,  all skew constacyclic codes
over a finite chain ring $\mathbb{F}_{p^m}+u\mathbb{F}_{p^m}$ are investigated.

This paper is organized as follows: Results concerning finite chain rings  and skew polynomials over these are
discussed in Section 2 along with some definitions and basic properties of skew constacyclic codes. In Section
3,  the algebraic structure and
 some properties of skew constacyclic codes whose generator polynomials are monic right divisors of $x^n-\lambda$ are
 established.  In many cases,   the structures of    their
 Euclidean and Hermitian
dual codes    are given. Necessary and sufficient conditions for
 them
to be  Euclidean and Hermitian self-dual
 are determined as well.
  A complete  structural classification of skew constacyclic codes over
$\mathbb{F}_{p^m}+u\mathbb{F}_{p^m}$ comes in Section~4.   Moreover, the structures of Euclidean and Hermitian
dual codes of skew cyclic and skew-negacyclic codes  are determined. Some illustration examples of   skew cyclic
codes are also provided.

\section{Preliminaries}
In this section, we recall and derive some useful results concerning finite chain rings and  skew polynomials
over such rings. The definition of a skew constacyclic code is introduced together with some basic properties.

\subsection{Finite Chain Rings}
 A finite commutative ring  with identity $1\neq0$ is called a
\textit{finite chain ring} if its ideals are linearly ordered by inclusion. It is known  that every ideal of a
finite chain ring is principal and its maximal ideal is unique (cf. \cite{Mc1974}). Throughout, let
$\mathcal{R}$ denote a finite chain ring, $\gamma$ a generator of its maximal ideal and $\mathcal{K}$ the
residue field $\mathcal{R}/\langle  \gamma\rangle$.   With these notations, the ideals of $\mathcal{R}$ form
the following chain
\[ \mathcal{R}= \langle  1\rangle \supsetneq \langle  \gamma\rangle \supsetneq \langle  \gamma^2 \rangle\supsetneq \dots \supsetneq \langle  \gamma^{e-1} \rangle \supsetneq \langle  \gamma^e\rangle=\langle  0\rangle .\]
The integer $e$ is called the \textit{nilpotency index} of $\mathcal{R}$. If $\mathcal{K}$ has $q$ elements,
then   $|\mathcal{R}|=q^e$.

Typical examples of finite chain rings which are not fields are the integer residue ring $\mathbb{Z}_{p^e}$, the
Galois ring ${\rm GR}(p^e,m)$ and the ring $\mathbb {F}_{p^m}+u \mathbb{F}_{p^m}+\dots+u^{e-1}\mathbb{F}_{p^m}$,
where $p$ is a prime number and $m,e$ are positive integers such that $e\geq 2$. Note that $\mathbb {F}_{p^m}+u
\mathbb{F}_{p^m}+\dots+u^{e-1}\mathbb{F}_{p^m}:= \{\sum_{i=0}^{e-1}a_iu^i\mid a_i\in \mathbb{F}_{p^m}\}$ is a
ring under the usual addition and multiplication of polynomials in indeterminate $u$ together with the rule
$u^e=0$. This ring is isomorphic to $\mathbb{F}_{p^m}[u]/\langle  u^e\rangle$, the only finite chain ring of
characteristic $p$, nilpotency index $e$, and residue field~$\mathbb{F}_{p^{m}}$ (cf. {\cite[Lemma
1]{ClLi1973}}). The reader can find   further details concerning finite chain rings in \cite{BiFl2002},
\cite{ClDr1973}, \cite{ClLi1973},
 \cite{Mc1974} and~\cite{Wa2003}.

In  \cite{Al-1991} and \cite{Al-1990}, the structure of the automorphism group $\Aut(\mathcal{R})$ of every
finite chain ring $\mathcal{R}$ has been characterized.  Many classes of finite chain rings have nontrivial
automorphism groups. For examples, $\Aut({\rm GR}(p^e,m))$ is non-trivial if and only if $m\geq 2$ (cf.
\cite{BoSoUm2008} and \cite{Wa2003}) and $\Aut(\mathbb {F}_{p^m}+u
\mathbb{F}_{p^m}+\dots+u^{e-1}\mathbb{F}_{p^m})$ is non-trivial if and only if $m\geq 2$ or $p$ is odd or $e\geq
3$ (cf. \cite[Proposition 1]{Al-1991}). Here, the automorphism group of $\mathbb{F}_{p^m}+u \mathbb{F}_{p^m}$ is
given as a corollary of \cite[Proposition 1]{Al-1991}, the complete characterization of the automorphism group
of $\mathbb {F}_{p^m}+u \mathbb{F}_{p^m}+\dots+u^{e-1}\mathbb{F}_{p^m}$.

\begin{corollary}[\cite{Al-1991}]  \label{cor-Auto}
For $\theta\in \Aut(\mathbb{F}_{p^m})$   and   $ \beta\in \mathbb{F}_{p^m}^*$, let \[\Theta_{\theta, \beta}:
\mathbb{F}_{p^m}+u \mathbb{F}_{p^m}\rightarrow \mathbb{F}_{p^m}+u \mathbb{F}_{p^m}\] be  defined by
\[\label{eq-auto2} \Theta_{\theta,\beta}(a+bu)= \theta(a)+\beta\theta(b)
u.
\]
Then  $\Aut(\mathbb{F}_{p^m}+u \mathbb{F}_{p^m})=\{\Theta_{\theta,\beta} \mid \theta\in \Aut(\mathbb{F}_{p^m})
\text{ and  } \beta\in \mathbb{F}_{p^m}^*\}$.
\end{corollary}
Note that $\mathbb{F}_{p^m}+u \mathbb{F}_{p^m}$ and its automorphisms play  an important role in later examples
and in Section 4.

\subsection{Skew Polynomial Rings over Finite Chain Rings}
In \cite{BoGe2007}, \cite{BoSoUm2008}, \cite{BoUl2009} and
 \cite{Mc1974}, results concerning skew polynomial rings over finite fields and
over Galois rings have been  studied. Applying the ideas   in these references, the following results over
finite chain rings  are given.

For a given automorphism $\Theta$  of $\mathcal{R}$, the set
  $\mathcal{R}[x;\Theta]=\{a_0+a_1x+\dots+a_nx^n\mid a_i\in \mathcal{R}\text{
and } n\in \mathbb{N}_0\}$ of formal polynomials  forms a ring under the usual addition of polynomials and where
the multiplication is defined using the rule
  $xa=\Theta(a)x$. The multiplication is extended to all elements
  in
$\mathcal{R}[x;\Theta]$
 by associativity  and distributivity.    The ring $\mathcal{R}[x;\Theta]$
 is  called a \textit{skew polynomial ring} over $\mathcal{R}$ and an element in
 $\mathcal{R}[x;\Theta]$ is called a \textit{skew polynomial}.
 It is easily   seen that the ring~$\mathcal{R}[x;\Theta]$ is
 non-commutative unless $\Theta$ is the identity automorphism on $\mathcal{R}$.

 Based on  the canonical reduction modulo $\gamma$, $\bar{ } : \mathcal{R}\rightarrow \mathcal{K}$,
 the automorphism $\bar{\Theta}$ of $\mathcal{K}$ is induced from
 $\Theta$ by \begin{align*} 
\bar{\Theta}(\bar{r}) =\overline{\Theta(r)} \text{ for all } r\in
 \mathcal{R}.
 \end{align*}
Then there is a natural   ring epimorphism extension $\bar{ } : \mathcal{R}[x;\Theta]\rightarrow
\mathcal{K}[x;\bar{\Theta}]$ defined by
\[r_0+r_1x+\dots+r_nx^n\mapsto \bar{r_0}+\bar{r_1}x+\dots+\bar{r_n}x^n.\]
In other words, for each $f(x)\in \mathcal{R}[x;\Theta]$, $\overline{f(x)}$ denotes the componentwise reduction
modulo $\gamma$ of $f(x)$.

The ring $\mathcal{R}[x; \Theta]$ does not need to be a unique factorization ring. Moreover, the degrees of the
irreducible factors are not unique up to permutation.
\begin{example}\label{ex2.2}
Consider  the automorphism $\Theta_{{\rm id},2}$ of $\mathbb{F}_{3}+u\mathbb{F}_{3}$, where $\Theta_{{\rm
id},2}(a+ bu) = a +2bu$. Here are two irreducible factorizations of $x^6-1$ in
$(\mathbb{F}_{3}+u\mathbb{F}_{3})[x; \Theta_{{\rm id},2} ]$
\begin{align*}
x^6-1&=(x+1)^3(x+2)^3\\
     &=(x^2+ux+2)^3.
\end{align*}
\end{example}

The skew polynomial ring $\mathcal{R}[x;\Theta]$ is neither left nor right Euclidean. However,   left and right
divisions can be defined for some suitable elements.  Let $f(x)=a_0+a_1x+\dots+a_rx^r$ and
$g(x)=b_0+b_1x+\dots+b_sx^s$, where $b_s$ is a unit in $ \mathcal{R}$.  The \textit{right division} of $f(x)$ by
$g(x)$ is defined as follows:

 If $r< s$, then $f(x)=0g(x)+f(x).$ Suppose that $r\geq s$. First, note that  the degree of
\[f(x)-a_r\Theta^{r-s}(b_s^{-1})x^{r-s}g(x)\]
is less than the degree of $f(x)$. Then iterating the above procedure  by subtracting further left multiples of
$g(x)$ from the result until the degree is less than the degree of $g(x)$, we obtain skew polynomials $q(x)$ and
$r(x)$ such that
 \[f(x)=q(x)g(x)+r(x)
\text{ with } \deg(r(x))<\deg(g(x)) \text{ or } r(x)=0 .\] Note~that~$q(x)$~and $r(x)$ are unique and called the
\textit{right~uotient}~and~\textit{right remainder}, respectively. This algorithm is called the \textit{Right
Division Algorithm}~in~$\mathcal{R}[x;\Theta]$.

If $r(x)=0$, we say that $g(x)$ is a \textit{right divisor} of $f(x)$. In this case,  denote by
$\displaystyle\frac{f(x)}{g(x)}$ the right quotient $q(x)$ of $f(x)$ by $g(x)$.   This implies
 \begin{align}
 f(x)=\displaystyle\frac{f(x)}{g(x)}g(x).\label{eq:frac-skewR}
 \end{align}

Similarly, the \textit{Left Division Algorithm} in $\mathcal{R}[x;\Theta]$   can be defined using the fact that
the degree~of
\[f(x)-g(x)\Theta^{-s}(a_rb^{-1}_s)x^{r-s}\] is less than the degree of $f(x)$.

For a skew polynomial $f(x)$ in $\mathcal{R}[x;\Theta]$, let $\langle f(x) \rangle$ denote the left ideal of
$\mathcal{R}[x;\Theta]$  generated by $f(x)$. Note that $\langle f(x) \rangle$ does not need to  be two-sided. A
sufficient condition for $\langle  f(x) \rangle$ to be two-sided is given as follows:

\begin{proposition}\label{Thm2.3}
 If $f(x)=x^tg(x)$ such that $g(x)$ is central and $t\in\mathbb{N}_0$, then   $\langle  f(x) \rangle$ is a principal two-sided
ideal in $\mathcal{R}[x;\Theta]$.
\end{proposition}
\begin{proof}
 Since $g(x)$ is
  central, for each skew polynomial $
\sum_{i=0}^n a_ix^i$ in $\mathcal{R}[x;\Theta]$, we have $
 \left(\sum_{i=0}^n a_ix^i\right)(x^tg(x))
                 =x^t\sum_{i=0}^n\Theta^{-t}(a_i)x^i g(x)=(x^tg(x))\sum_{i=0}^n\Theta^{-t}(a_i)x^i.
$
\end{proof}

\begin{corollary}\label{cor2.4}
If $f(x)$ is a monic central skew polynomial  of degree $n$, then the  skew polynomials of degree less than $n$
are canonical representatives of the elements in $\mathcal{R}[x,\Theta]/\langle f(x)\rangle$.
\end{corollary}
 \begin{proof}
 By Proposition \ref{Thm2.3}, $\langle  f(x) \rangle$ is  two-sided
  and hence the quotient $\mathcal{R}[x,\Theta]/\langle
 f(x)\rangle$ is meaningful. The  desired result follows from the Right
 Division Algorithm.
 \end{proof}

\begin{proposition}\label{prob2.5} Let $n$ be a positive integer and
$\lambda$ a unit in $\mathcal{R}$.  Then the following statements are equivalent:
\begin{enumerate}[$i)$]
\item   $x^n-\lambda$ is central in $\mathcal{R}[x,\Theta]$.%
\item   $\langle x^n-\lambda\rangle$ is
two-sided.%
\item   $n$ is a multiple of the order of $\Theta$ and    $\lambda$ is   fixed by $\Theta$.%
\end{enumerate}
\end{proposition}
\begin{proof}
  $i)\Rightarrow ii)$ follows directly from  Proposition \ref{Thm2.3}.

  Next, we prove $ii) \Rightarrow iii)$. Assume that $\langle x^n-\lambda\rangle$ is two-sided. Let $r\in \mathcal{R}$.
 Then $rx^n-r\lambda=r(x^n-\lambda)=(x^n-\lambda)s=\Theta^n(s)x^n-s\lambda$ for some $s\in \mathcal{R}$.
  Comparing the coefficients, we have $r\lambda=s\lambda$. As
 $\lambda$ is a unit, it follows that  $r=s$, and hence $ rx^n-r\lambda= \Theta^n(r)x^n- r\lambda.$ Thus, $n$ is a multiple of the
order of $\Theta$.  Next, we observe that
$x^{n+1}-\Theta(\lambda)x=x(x^n-\lambda)=(x^n-\lambda)(ax+b)=\Theta^n(a)x^{n+1}+\Theta^n(b)x^n-a\lambda
x-b\lambda$,
 for some   $a$ and $b$ in $\mathcal{R}$. Then $\Theta^n(a)=1$ and $\Theta^n(b)=0$. As $\Theta$ is an
 automorphism, we have $a=1$ and $b=0$, and hence $x^{n+1}-\Theta(\lambda)x = x^{n+1} - \lambda x  .$ Therefore,
  $\lambda$ is fixed by $\Theta$.

Finally, we prove $iii)\Rightarrow i)$. Assume that $n$ is a multiple of the order of  $\Theta$ and $\lambda$ is
fixed by $\Theta$.  Then $x(x^n-\lambda) = x^{n+1}- \Theta(\lambda)x=
 x^{n+1}- \lambda x=  (x^n-\lambda)x$ and  $(x^n-\lambda)t=\Theta^n(t)x^n-t\lambda= t x^n-t\lambda= t(x^n-\lambda)$, for all
$t\in \mathcal{R}$. Consequently, $x^n-\lambda$ commutes with any skew polynomial in $\mathcal{R}[x;\Theta]$.
\end{proof}

\begin{proposition}\label{prop2.6} Let $h(x),g(x)\in
\mathcal{R}[x;\Theta]$. If $h(x)g(x)$ is a monic central skew polynomial, then $h(x)g(x)=g(x)h(x)$. In
particular, if $g(x)$ is a right divisor of a monic central skew polynomial $f(x)$, then $g(x)$ and the right
quotient $\displaystyle\frac{f(x)}{g(x)}$ commute, i.e.
\begin{align}
g(x)\displaystyle\frac{f(x)}{g(x)}=f(x)=\displaystyle\frac{f(x)}{g(x)}g(x).\label{eq:frac-skew}
 \end{align}
\end{proposition}
\begin{proof} Assume that $h(x)g(x)$ is monic and  central. Then the leading
coefficient of $g(x)$ and $h(x)$ are units. Since $h(x)g(x)$ is central, we have
$$h(x)(h(x)g(x))=(h(x)g(x))h(x)=h(x)(g(x)h(x)).$$ Thus
$h(x)(h(x)g(x)- g(x)h(x))=0$. As the leading coefficient of $h(x) $ is a unit, $h(x)$ is not a zero divisor.
Hence $h(x)g(x)= g(x)h(x)$ as desired.
\end{proof}

The following discussion guarantees the existence of the right localization of $\mathcal{R}[x;\Theta]$ which
plays a vital role in the study of dualities of codes. In the light of \cite[Theorem 2]{Ri1970},  necessary and
sufficient conditions for $\mathcal{R}[x;\Theta]$ to have the right localization are given as follows.
\begin{theorem}[{\cite{Ri1970}}]\label{thm:localization}
Let $S=\{x^i\mid i\in\mathbb{N}\}$. Then $\mathcal{R}[x;\Theta]$ has the right localization at $S$ if and only
if both the following conditions hold:
\begin{enumerate}[$i)$] \item
For all $x^{i}\in S$ and $a(x)\in \mathcal{R}[x;\Theta]$, there exist $x^{j}\in S$ and $b(x)\in
\mathcal{R}[x;\Theta]$ such
that $a(x) x^{i} = x^{j} b(x)$.  %
\item Given $a(x)\in \mathcal{R}[x;\Theta]$ and $x^{i}\in S$, if         $x^{i}a(x) = 0$, then there exists
$x^{j}\in S$ such that $a(x) x^{j }= 0.$
\end{enumerate}
\end{theorem}
\noindent Condition $i)$ holds because the multiplication rule allows the shifting of powers of $x$ from left to
right by changing the coefficients.  Note that, for each $x^{i}\in S$,  it is never a left zero divisor. If
$a(x)\in \mathcal{R}[x;\Theta]$ such that $x^{i}a(x) = 0$, then $a(x)$ must be zero and hence $a(x) x^{j }= 0,$
for all $x^j\in S$. This obviously implies $ii)$. Then, by Theorem~\ref{thm:localization}, the right
localization $\mathcal{R}[x;\Theta]S^{-1}$ of $\mathcal{R}[x;\Theta]$ exists. We have $ax^{-1} =
x^{-1}\Theta(a)$ where $x^{-1}$ is the inverse of $x$ in this right localization.

 The following map is  key to
determining the structure of dual codes.

\begin{proposition}\label{prop2.7}
Let $\varphi: \mathcal{R}[x;\Theta] \rightarrow \mathcal{R}[x;\Theta]S^{-1}$ be   defined by
\begin{align*}\label{eq-InProp2.7}
 \varphi(\sum_{ i=0}^t a_ix^{i} )= \sum_{
i=0}^t x^{-i}a_i.
\end{align*}  Then $\varphi$ is a ring anti-monomorphism.
\end{proposition}
\begin{proof}
The proof is similar to a part of the argument used in the proof of \cite[Theorem~4.4]{BoSoUm2008}.
\end{proof}

\subsection{Definitions and Basic Properties of Skew Constacyclic Codes over
Finite Chain Rings}

  A \textit{code  of
length} $n$ over   $ {\mathcal{R}}$ is a nonempty subset of ${\mathcal{R}}^n$. A code ${C}$ is said to be
\textit{linear} if it is a submodule of the ${\mathcal{R}}$-module ${\mathcal{R}}^n$. In this paper, all codes
are assumed to be linear unless otherwise stated.

Given an automorphism $\Theta$ of  $\mathcal{R}$ and  a unit $\lambda$ in ${\mathcal{R}} $, a
  code ${C}$ is said to be \textit{skew constacyclic}, or specifically,
$\Theta$-$\lambda$-\textit{constacyclic} if
$C$ is closed under the $\Theta$-$\lambda$-\textit{constacyclic shift}
$\rho_{\Theta,\lambda}:\mathcal{R}^n\rightarrow \mathcal{R}^n$ defined by
$$\rho_{\Theta,\lambda}((a_0,a_1,\dots,a_{n-1}) )= (\Theta(\lambda
a_{n-1}),\Theta(a_0),\dots,\Theta(a_{n-2})).$$ In particular,
  such codes are called
\textit{skew  cyclic} and \textit{skew negacyclic codes} when $\lambda$ is $1$ and $-1$, respectively. When
$\Theta$ is the identity automorphism, they become classical constacyclic, cyclic and negacyclic codes.

Analogous to the  classical constacyclic codes, we characterize $\Theta$-$\lambda$-constacyclic codes in terms
of left ideals in $\mathcal{R}[x;\Theta]/\langle x^n-\lambda\rangle$. However, due to Proposition \ref{prob2.5},
$\mathcal{R}[x;\Theta]/\langle x^n-\lambda\rangle$ is meaningful if only if $\langle x^n-\lambda\rangle$ is
two-sided, or equivalently, $n$ is a multiple of the order of $\Theta$ and $\lambda$ is a unit fixed
by~$\Theta$.

For this purpose, throughout, we restrict our study to  the case where the length $n$ of codes is a multiple of
the order of $\Theta$ and $\lambda$ is a unit in $\mathcal{R}^\Theta$, where $\mathcal{R}^\Theta$ denotes the
subring of $\mathcal{R}$   fixed by $\Theta$.

The \textit{skew polynomial representation} of a code $C$ is defined to be $\{c_0+c_1x+\dots+c_{n-1}x^{n-1} \mid
(c_0,c_1,\dots,c_{n-1})\in C \}$. For convenience, it will be regarded as $C$   itself. The next theorem is
analogous to that for classical constacyclic codes. The proof is omitted.
  \begin{theorem}\label{thm2.8}
  A   code $C$ of length
  $n$ over $\mathcal{R}$  is $\Theta$-$\lambda$-constacyclic
  if and only if   the skew polynomial representation of $  C $ is a left
  ideal in $ \mathcal{R} [x;\Theta]/\langle  x^n-\lambda\rangle$.
\end{theorem}

There are two inner products on $\mathcal{R}^n$ that  we are interested in. One is the \textit{Euclidean inner
product} defined by  $\langle  {u}, {v}\rangle=\sum_{i=0}^{n-1}u_iv_i,$ for $u=(u_0,u_1,\dots,u_{n-1})$ and
$v=(v_0,v_1,\dots,v_{n-1})$ in $\mathcal{R}^n$. When the order of $\Theta$ is $2$, we can also consider the
\text{Hermitian inner product} which is defined as $\langle {u}, {v}\rangle_H=\sum_{i=0}^{n-1}u_i\Theta(v_i).$
Vectors $u$ and $v$ are said to be \textit{Euclidean orthogonal} (resp., \textit{Hermitian orthogonal}) if
$\langle u,v\rangle=0$ (resp., $\langle u,v\rangle_H=0$).

 The \textit{Euclidean dual code} of a code $C$ of
length $n$ over $\mathcal{R}$ is defined to be $C^\perp=\{v\in \mathcal{R}^n\mid \langle v,c\rangle =0 \text{
for all } c\in C\}.$ Similarly, the \textit{Hermitian dual code}  of   $C$   is defined as $C^{\perp_H}=\{v\in
\mathcal{R}^n\mid \langle v,c\rangle_H =0 \text{ for all } c\in C\}.$ The code $C$ is said to be
\textit{Euclidean self-dual} (resp., \textit{Hermitian self-dual}) if $C=C^\perp$ (resp., $C=C^{\perp_H}$).


\section{Skew Constacyclic Codes Generated  by   Monic  Right Divisors  of $x^n-\lambda$}

 In this section, we focus on   $\Theta$-$\lambda$-constacyclic codes  which are principal
left ideals  in  $\mathcal{R}[x;\Theta]/\langle
 x^n-\lambda\rangle$ generated by monic right divisors of
 $x^n-\lambda$. We derived some useful tools  and extend   results on constacyclic codes over Galois rings   \cite[Sections 4-5,7]{BoSoUm2008}
  to the case over an arbitrary finite chain
 ring~$\mathcal{R}$.    The main  assumptions  that $\lambda$ is a unit  in $\mathcal{R}^\Theta$
 and the length $n$ of codes is a multiple of the order of $\Theta$ are
 assumed.

\subsection{Properties of Skew Constacyclic Codes Generated  by   Monic  Right Divisors  of $x^n-\lambda$}


Given a right divisor $g(x)=\sum_{i=0}^{n-k-1}g_ix^i+x^{n-k}$ of
 $x^n-\lambda$, a generator matrix of the $\Theta$-$\lambda$-constacyclic code $C$ generated by $g(x)$ is given by
\[ G =   \left[%
\begin{array}{  ccccccc}
  g_0             & \dots                     & g_{n-k-1} &1& 0                            & \dots     & 0   \\
  0         & \Theta(g_{0})                             &      \dots         & \Theta(g_{n-k-1})&1                   & \dots      & 0  \\
  0                              & \dots              &         \dots        &    \dots                 &  \Theta^2(g_{n-k-1})   & \ddots     & 0   \\
  \vdots    &   \ddots                &         \ddots        &    \ddots                 & \ddots    & \ddots     & \vdots  \\
 0                    &   \dots                           &      0       & \Theta^{k-1}(g_0)     & \dots                     & \Theta^{k-1}(g_{n-k-1})&1   \\
\end{array}%
\right] .\] Then the rows of $G$ are linearly independent, and hence the next proposition follows.
\begin{proposition}\label{prop:free}
Let $g(x)$ be a right divisor of $x^n-\lambda$. Then  the $\Theta$-$\lambda$-constacyclic code $C$ generated by
$g(x)$ is
 a free $\mathcal{R}$-module with $|C|=|\mathcal{R}|^{n-\deg(g(x))}$.
\end{proposition}
When  $\Theta$ is the identity automorphism,  a $\Theta$-$\lambda$-constacyclic code becomes
 $\lambda$-constacyclic. However, the converse does not need to be true.   Here, necessary and sufficient conditions for a $\Theta$-$\lambda$-constacyclic code generated by
a right divisor of $x^n-\lambda$  to be $\lambda$-constacyclic are given.
\begin{proposition}\label{lem3.1}  Let
$g(x)$ be a monic right divisor of $x^n-\lambda$ in $\mathcal{R}[x;\Theta]$. The
 $\Theta$-$\lambda$-constacyclic code generated by $g(x)$ is
$\lambda$-constacyclic if and only if $g(x)\in \mathcal{R}^\Theta[x;\Theta]$.
\end{proposition}

\begin{proof}
Suppose $g(x) =\sum_{i=0}^{n-k-1}g_ix^i+x^{n-k}$  and $C$ is the $\Theta$-$\lambda$-constacyclic code generated
by $g(x)$.

Assume that $C$ is $\lambda$-constacyclic. Then   $xg(x), g(x)x\in C$. By the linearity of $C$, $xg(x)-g(x)x\in
C$ and hence
\[(\Theta(g_0)-g_0)x+(\Theta(g_1)-g_1)x^2+\dots+(\Theta(g_{n-k-1})-g_{n-k-1})x^{n-k}=p(x)g(x),\]
for some $p(x)\in \mathcal{R}[x;\Theta]$ such that $\deg(p(x))<k$. Thus $\deg(p(x)g(x))<n$ which implies that
$p(x)$ is constant such that $p(x)g_0=0$. Since $g(x)$ is a right divisor of $x^n-\lambda$ and $\lambda$ is a
unit, $g_0$ is not a zero divisor. Thus $p(x)$ is zero and hence $g_i$ is fixed by $\Theta$ for all $i$.

Conversely, assume that $g(x)\in \mathcal{R}^\Theta[x;\Theta]$. Then $g_ix=xg_i$ for all $i=0,1,\dots,{n-k}.$
Thus $g(x)x=xg(x)\in C$, therefore, the result follows.
\end{proof}

A parity-check matrix for $C$ is determined in the next proposition.

\begin{proposition}\label{prop3.2}  Let $C$ be the
$\Theta$-$\lambda$-constacyclic code  generated by a monic right divisor
 $g(x)$ of $x^n-\lambda$ and $h(x):=\displaystyle\frac{x^n-\lambda}{g(x)}$.
 Then the following statements hold:
\begin{enumerate}[$i)$]
\item For $c(x)\in\mathcal{R}[x;\Theta]$,  $c(x)\in C$ if and only if $c(x)h(x)=0$ in
$\mathcal{R}[x;\Theta]/\langle x^n-\lambda\rangle$.%
\item   If $h(x)=\sum_{i=0}^{k-1}h_ix^i+x^k$, then the following matrix
\[ H= \left[%
\begin{array}{ ccccccc}
  1         & \Theta(h_{k-1})      & \dots                     & \Theta^k(h_0) & 0                              & \dots     & 0   \\
  0         & 1                       & \Theta^2(h_{k-1})         &      \dots         & \Theta^{k+1}(h_0)                  & \dots      & 0  \\
  0    & 0                       & \dots              &         \dots        &    \dots                   & \ddots     & 0   \\
  \vdots    & \vdots                    & \ddots                &         \ddots        &    \ddots                   & \ddots      & \vdots  \\
 0         & 0            & \dots                           & 1             & \Theta^{n-k}(h_{k-1})                & \dots     & \Theta^{n-1}(h_0)   \\
\end{array}%
\right]\] is a parity-check matrix for $C$.%
\end{enumerate}
\end{proposition}
\begin{proof} Since $n$ is a multiple of the order of $\Theta$ and $\lambda\in \mathcal{R}^\Theta$,
$x^n-\lambda$ is central  and    it follows from Proposition~\ref{prop2.6} that $x^n-\lambda=h(x)g(x)=g(x)h(x)$.

First, we prove $i)$.
  Assume that $c(x)=p(x)g(x)$ for some $p(x)$ in $\mathcal{R}[x;\Theta]$.
  Then
$ c(x)h(x) =(p(x)g(x))h(x)  =p(x) (x^n-\lambda) =0 $ in $ \mathcal{R}[x;\Theta]/\langle  x^n-\lambda\rangle.$

Conversely, assume that $c(x)h(x)=0$ in $\mathcal{R}[x;\Theta]/\langle  x^n-\lambda\rangle $. Then there exists
$p(x)\in \mathcal{R}[x;\Theta]$ such that $ c(x)h(x) =p(x) (x^n-\lambda) =p(x)g(x)h(x) . $ As the leading
coefficient of $h(x)$ is a unit, we then have $c(x)=p(x)g(x)\in C$.

To prove $ii)$, let $c(x)=c_0+c_1x+\dots+c_{n-1}x^{n-1}\in C$ and let
$[\,s_k~s_{k+1}~\cdots~s_{n-1}\,]=[\,c_0~c_1~\cdots~c_{n-1}\,]H^T.$ Then, for $l\in\{k,k+1,\dots, n-1\}$,
\[s_l=c_{l-k}+\sum_{j=0}^{k-1}c_{l-j}\Theta^{l-j}(h_j)\]
which equals  the coefficient of $x^l$ in $c(x)h(x)$.

 Since $c(x)\in C$, it follows from
   $i)$ that    $c(x)h(x)=0$ in
$\mathcal{R}[x;\Theta]/\langle  x^n-\lambda\rangle$.
 Then there exists
$q(x)\in \mathcal{R}[x;\Theta]$ such that $q(x)(x^n-\lambda)=c(x)h(x)$ having degree less than $n+k$. Therefore,
the coefficients of the monomials $x^k,x^{k+1}, \dots , x^{n-1}$ in this product must be zero, i.e.,
$[\,s_k~s_{k+1}~\cdots~s_{n-1}\,]$ is the zero matrix.

  Since the rank of $H$ is $n-k$, the result
follows.
\end{proof}


\subsection{Euclidean Dual Codes}
We  study Euclidean dual codes of $\Theta$-$\lambda$-constacyclic codes over $\mathcal{R}$. Their
characterization  is given. When $\lambda^2=1$, a generator   of the Euclidean dual code of a
 $\Theta$-$\lambda$-constacyclic code  is determined. Necessary and sufficient conditions for
 such a code
to be  Euclidean self-dual
 are given as well.

\begin{lemma}\label{lem2.pre9}  Let
$C$ be a code of length $n$ over $\mathcal{R}$.  Then $C$ is $\Theta$-$\lambda$-constacyclic  if and only if
$C^\perp$ is $\Theta$-$\lambda^{-1}$-constacyclic. In particular, if $\lambda^2=1$, then   $C$ is
$\Theta$-$\lambda$-constacyclic if and only if $C^\perp$ is $\Theta$-$\lambda$-constacyclic.
\end{lemma}
\begin{proof} Note that,  for each unit  $\alpha$ in $\mathcal{R}$, $\alpha\in \mathcal{R}^\Theta$
if and only if $\alpha^{-1}\in\mathcal{R}^\Theta$.   Since $\lambda\in  \mathcal{R}^\Theta$, so is
$\lambda^{-1}$. Let $u=(u_0,u_1,\dots,u_{n-1})\in C$ and $v=(v_0,v_1,\dots,v_{n-1})\in C^\perp$. Since
$(\Theta^{n-1}(\lambda u_1), \Theta^{n-1}(\lambda u_2),\dots,\Theta^{n-1}(\lambda u_{n-1}),\Theta^{n-1}( u_0)
)=\rho_{\Theta,\lambda}^{n-1}(u)\in C$, we have
\begin{align*}
  0&=\langle \rho_{\Theta,\lambda}^{n-1}(u),v\rangle\\
   &=\langle (\Theta^{n-1}(\lambda u_1), \Theta^{n-1}(\lambda
u_2),\dots,\Theta^{n-1}(\lambda u_{n-1}),\Theta^{n-1}( u_0) ),(v_0,v_1,\dots,v_{n-1})
  \rangle\\
  &=\lambda\langle (\Theta^{n-1}( u_1), \Theta^{n-1}(
u_2),\dots,\Theta^{n-1}(  u_{n-1}),\Theta^{n-1}( \lambda^{-1}u_0) ),(v_0,v_1,\dots,v_{n-1})
  \rangle\\
  &=\lambda (\Theta^{n-1}( \lambda^{-1}u_0)v_{n-1}+ \sum_{i=1}^{n-1} \Theta^{n-1}( u_i)v_{i-1}
  ).
\end{align*} As $n$ is a multiple of the order of $\Theta$ and $\lambda^{-1}$ is fixed by $\Theta$, it
follows that
\begin{align*}
0=\Theta(0)&=\Theta(\lambda (\Theta^{n-1}(
                \lambda^{-1}u_0)v_{n-1}+ \sum_{i=1}^{n-1} \Theta^{n-1}(
                u_i)v_{i-1}  ))\\
            &= \lambda (
                u_0\Theta (\lambda^{-1}v_{n-1})+ \sum_{i=1}^{n-1}
                u_i\Theta (v_{i-1})  ) \\
            &=\lambda\langle \rho_{\Theta,{\lambda}^{-1}} (v),u\rangle.
\end{align*}
Therefore, $\rho_{\Theta,\lambda^{-1}} (v)\in C^\perp$.

The converse follows from the fact that  $(C^\perp)^\perp=C$.

In addition, assume that $\lambda^2=1$. Then $\lambda=\lambda^{-1}$ and hence the last statement follows
 immediately from the main result.
\end{proof}

If $\lambda^2=1$, it follows from the previous lemma that the Euclidean dual  $C^\perp$ of  a
$\Theta$-$\lambda$-constacyclic code $C$ is again $\Theta$-$\lambda$-constacyclic. In this case, a generator of
$C^\perp$  is  given through  the ring anti-monomorphism $\varphi$ defined in Proposition \ref{prop2.7}, where $
 \varphi(\sum_{ i=0}^t a_ix^{i} )= \sum_{
i=0}^t x^{-i}a_i$. The next lemma is key to obtaining this result.

\begin{lemma}\label{lem2.9} Assume that $\lambda^2=1$. Let $a(x)=a_0+a_1x+\dots+a_{n-1}x^{n-1}$ and
$b(x)=b_0+b_1x+\dots+b_{n-1}x^{n-1}$ be in $\mathcal{R}[x;\Theta]$. Then the following statements are
equivalent:
\begin{enumerate}[$i)$]
\item The coefficient vector of $a(x)$ is Euclidean orthogonal to the coefficient vector of $x^i
(x^{n-1}\varphi(b(x)))$ for all $i\in\{0,1,\dots,n-1\}$.%
\item $(a_0,a_1,\dots,a_{n-1})$ is Euclidean orthogonal to $(b_{n-1},\Theta(b_{n-2}),\dots, \Theta^{n-1}(b_0))$
and all its
$\Theta$-$\lambda$-constacyclic shifts.%
\item $a(x)b(x)=0$ in $\mathcal{R}[x;\Theta]/\langle x^n-\lambda\rangle$.%
\end{enumerate}
\end{lemma}
\begin{proof}
$i)$ if and only if $ii)$ follows directly from the definition of $\varphi$. We prove $ii)$ if and only if
$iii)$.  Let $a(x)b(x)=c_0+c_1x+\dots +c_{n-1}x^{n-1}\in \mathcal{R}[x;\Theta]/\langle x^n-\lambda\rangle$.
Since $\lambda\in \mathcal{R}^\Theta$ such that $\lambda^2=1$ and $n$ is a multiple of the order of $\Theta$, it
follows that,  for each $k\in\{0,1,\dots,n-1\}$,
\begin{align*}
c_k&=\sum_{\substack{i+j=k\\ 0\leq i\leq n-1\\ 0\leq j\leq n-1}} a_i\Theta^i(b_j)+\sum_{\substack{i+j=k+n\\
0\leq i\leq n-1\\ 0\leq
j\leq n-1}} \lambda a_i\Theta^i(b_j)\\
&=\lambda\left(\sum_{\substack{i+j=k\\ 0\leq i\leq n-1\\ 0\leq j\leq n-1}}
a_i\Theta^{k-j}(\lambda b_j)+\sum_{\substack{i+j=k+n\\ 0\leq i\leq n-1\\
0\leq
j\leq n-1}}  a_i\Theta^{n+k-j}(b_j)\right)\\
&=\lambda\langle (a_0,a_1,\dots,a_{n-1}), \\
&\text{ }\text{ }\text{ }\text{ }\text{ }\text{ }\text{ }  (\lambda  b_{k } ,\Theta
(\lambda b_{k-1}),\dots, \Theta^{k}(\lambda b_{0}),\Theta^{k+1}(b_{n-1}) ,\dots ,\Theta^{n-1}(b_{k+1}))\rangle\\
&=\lambda \langle (a_0,a_1,\dots,a_{n-1}),(\Theta^{(n-k)+k}(\lambda b_{k}),\Theta^{(n-k+1)+k}(\lambda b_{k-1}), \dots,\\
&\text{ }\text{ }\text{ }\text{ }\text{ }\text{ }\text{ }  \Theta^{k}(\lambda b_{0}),\Theta^{1+k}(b_{n-1})
,\dots ,\Theta^{(n-k-1)+k}(b_{k+1}))\rangle.
\end{align*}
Hence, $a(x)b(x)=0$ if and only if $c_k=0$ for all $k\in\{0,1,\dots,n-1\}$, which is true if and only if
$(a_0,a_1,\dots,a_{n-1})$ is Euclidean orthogonal to $(b_{n-1},\Theta(b_{n-2}),\dots, \Theta^{n-1}(b_0))$ and
all its $\Theta$-$\lambda$-constacyclic shifts.
\end{proof}

\begin{theorem}\label{thm3.4}  Assume that
$\lambda^2=1$. Let $g(x)$ be a right divisor of  $x^n-\lambda$ and
$h(x):=\displaystyle\frac{x^n-\lambda}{g(x)}$.  Let $C$ be  the $\Theta$-$\lambda$-constacyclic code generated
by $g(x)$. Then the following statements hold:
\begin{enumerate}[$i)$]
 \item The skew polynomial $x^{\deg(h(x))}\varphi(h(x))$ is a right divisor of $x^n-\lambda$.%
 \item The Euclidean dual $C^\perp$ is  a
$\Theta$-$\lambda$-constacyclic code generated by $x^{\deg(h(x))}\varphi(h(x))$.
\end{enumerate}
\end{theorem}
\begin{proof}
First, we prove $i)$.   Using the assumptions that $ n$   is a multiple of the order of~$\Theta$  and
$\lambda\in\mathcal{R}^\Theta $, we observe that
\begin{align*}
  \left(\varphi(g(x))(-\lambda) x^{n-{\deg(h(x))}}\right)\left(x^{\deg(h(x))}\varphi(h(x))\right)
  &=\varphi(g(x))(-\lambda)  x^{n }\varphi(h(x))\\
  &=-\lambda x^{n } \varphi(g(x))\varphi(h(x)) \\
  &=-\lambda x^{n }\varphi(h(x)g(x)), \\
  &(\text{since  } \varphi \text{ is a ring anti-monomorphism}) \\
  &=-\lambda x^{n }\varphi( x^n-\lambda) \\
  &=-\lambda x^{n } ( x^{-n}-\lambda) \\
  &=x^n- {\lambda}.
\end{align*}
As  $\varphi(g(x))(-\lambda) x^{n-{\deg(h(x))}}$ and $x^{\deg(h(x))}\varphi(h(x))$ belong to
$\mathcal{R}[x;\Theta]$, $x^{\deg(h(x))}\varphi(h(x))$ is a right divisor of $x^n-\lambda$ in
$\mathcal{R}[x;\Theta]$.


Next, we prove $ii)$. Since $g(x)h(x)=x^n-\lambda=0$ in $ \mathcal{R}[x;\Theta]/ \langle x^n-\lambda \rangle $,
by Lemma~\ref{lem2.9}, $\langle x^{\deg(h(x))}\varphi(h(x)) \rangle\subseteq C^\perp$. As
$x^{\deg(h(x))}\varphi(h(x))$ is a right divisor of $x^n-\lambda$, by Proposition \ref{prop:free},
 $|\langle x^{\deg(h(x))}\varphi(h(x))
\rangle|=|\mathcal{R}|^{n-\deg(h(x))}=|C^\perp|$. Therefore, $\langle x^{\deg(h(x))}\varphi(h(x)) \rangle=
C^\perp$.
\end{proof}

Necessary and sufficient conditions for a $\Theta$-$\lambda$-constacyclic code to be Euclidean self-dual
 are given in the next theorem.

\begin{theorem}\label{thm3.5} Assume that $\lambda^2=1$ and $n$ is even, denoted by $n=2k$.
Let $g(x)= \sum_{i=0}^{k-1}g_ix^i+x^{k}$ be  a right divisor of $x^n-\lambda$. Then the
$\Theta$-$\lambda$-constacyclic code generated by $g(x)$ is Euclidean self-dual if and only if
\begin{align} \label{exx}(\sum_{i=0}^{k-1}g_ix^i+x^{k})(\Theta^{-k}(g_0^{-1})+\sum_{i=1}^{k-1}\Theta^{i-k}(g_0^{-1}g_{k-i})x^i+x^{k})=x^{n}-\lambda.
\end{align}
\end{theorem}
\begin{proof} Let $C$ be the $\Theta$-$\lambda$-constacyclic code
generated by $g(x)$ and let $g^\perp(x)$ be the generator polynomial of the Euclidean dual code $C^\perp$.
 Denote by
$h(x):=\sum_{i=0}^{k-1}h_ix^i+x^k$ the right quotient $\displaystyle\frac{x^n-\lambda}{g(x)}$.
 It
follows from Theorem~\ref{thm3.4} that
\begin{align}\label{eq:dual1}
g^\perp(x)=x^k\varphi(h(x))=\Theta^k(h_0)x^k+\dots+ \Theta(h_{k-1})x + 1 .
\end{align}

First, assume that $C$ is Euclidean self-dual.  It is easily seen
 that $g(x)$ is the unique monic generator of minimal degree
in $C$. Then $g^\perp(x)$ is a scalar multiple of $g(x)$ of the form
\begin{align}\label{eq:dual2}
g^\perp(x)=\Theta^k(h_0)g(x)=\Theta^k(h_0)(\sum_{i=0}^{k-1}g_ix^i+x^{k}) .
\end{align}
Comparing the coefficients in (\ref{eq:dual1}) and (\ref{eq:dual2}), we obtain
 $\Theta^k(h_0)g_0=1$  and $\Theta^k(h_0)g_i=\Theta^i(h_{k-i})$, for all
$i=1,2,\dots,k-1$.   Consequently,   $h_0=\Theta^{-k}(g_0^{-1})$ and
 $h_i=\Theta^i(h_0)\Theta^{i-k}(g_{k-i})= \Theta^{i-k}(g_0^{-1})\Theta^{i-k}(g_{k-i})=\Theta^{i-k}(g_0^{-1}g_{k-i}) $, for all $i=1,2,\dots,$ $k-1.$ and
$h(x)=\Theta^{-k}(g_0^{-1})+\sum_{i=1}^{k-1} \Theta^{i-k}(g_0^{-1}g_{k-i})x^i+x^k.$ Therefore,
  (\ref{exx}) holds.

Conversely, assume that (\ref{exx}) holds. Then $$h(x)=\Theta^{-k}(g_0^{-1})+
\sum_{i=1}^{k-1}\Theta^{i-k}(g_0^{-1}g_{k-i})x^i+x^{k}.$$ Hence, by Theorem~\ref{thm3.4},
\begin{align*}
g^\perp(x)=x^k\varphi(h(x))&= \sum_{i=1}^{k } (g_0^{-1}g_{ i})x^i+1= g_0^{-1}g(x).
\end{align*} This completes the proof.
\end{proof}
\begin{remark} From Theorem \ref{thm3.5}, we observe that  if there is a Euclidean self-dual $\Theta$-$\lambda$-constacyclic  code,
then $-\lambda=g_0\Theta^{-k}(g_0^{-1})=\Theta^{k}(g_0)g_0^{-1}$. Thus, if the order of $\Theta$ divides $k$ and
$\lambda\neq -1$, then there are no Euclidean self-dual $\Theta$-$\lambda$-constacyclic  codes of length $2k$.
In particular, if $\Theta$ is the identity automorphism and $\lambda\neq -1$, then there are no Euclidean
self-dual $\Theta$-$\lambda$-constacyclic codes of any length.
\end{remark}

\subsection{Hermitian Dual   Codes}
Due to the constraint in the definition of the Hermitian inner product, the Hermitian dual codes of skew
constacyclic codes are studied only when the order of $\Theta$ is $2$.  Using arguments similar to those in the
previous proofs,
  the following results concerning the Hermitian duality   are
obtained.
\begin{lemma}\label{lem2.pre10}
 Let  $C$ be a code of even length $n$ over
$\mathcal{R}$. Assume that the order  of $\Theta$ is  $2$. Then $C$ is $\Theta$-$\lambda$-constacyclic if and
only if $C^{\perp_H}$ is   $\Theta$-$\lambda^{-1}$-constacyclic. In particular, if $\lambda^2=1$, then  $C$ is
$\Theta$-$\lambda$-constacyclic if and only if $C^{\perp_H}$ is $\Theta$-$\lambda$-constacyclic.
\end{lemma}


When $\lambda^2=1$, a generator of the Hermitian dual code of a $\Theta$-$\lambda$-constacyclic code
  is determined through the ring anti-monomorphism
$\varphi$ defined in Proposition \ref{prop2.7} and a ring automorphism $\phi$ on $\mathcal{R}[x;\Theta]$ defined
by
\begin{align}\label{eq-lem-phi}
\phi(\sum_{ i=0}^t a_ix^{i} ) =\sum_{ i=0}^t \Theta(a_i)x^{ i}.
\end{align}

\begin{lemma}\label{lem2.10}  Assume that   the
order of $\Theta$ is $2$ and $\lambda^2=1$. Let $a(x)=a_0+a_1x+\dots+a_{n-1}x^{n-1}$ and
$b(x)=b_0+b_1x+\dots+b_{n-1}x^{n-1}$ be in $\mathcal{R}[x;\Theta]$.  Then the following statements are
equivalent:
\begin{enumerate}[$i)$]
\item The coefficient vector of $a(x)$ is Hermitian orthogonal to the coefficient vector of  $x^i
\phi(x^{n-1} \varphi(b(x)))$  for all $i\in\{0,1,\dots,n-1\}$.%
\item $(a_0,a_1,\dots,a_{n-1})$ is Hermitian orthogonal to $(\Theta^{-1}(b_{n-1}),  b_{n-2} ,\dots,
\Theta^{n-2}(b_0))$ and all its
$\Theta$-$\lambda$-constacyclic shifts.%
\item $a(x)b(x)=0$ in $ \mathcal{R}[x;\Theta]/\langle x^n-\lambda\rangle$.%
\end{enumerate}
\end{lemma}


\begin{theorem}\label{thm3.6}
Assume that the order of $\Theta$ is $2$ and $\lambda^2=1$. Let $g(x)$ be a right divisor of $x^n-\lambda$ and
$h(x):=\displaystyle\frac{x^n-\lambda}{g(x)}$. Let  $C$ be the $\Theta$-$\lambda$-constacyclic code generated by
$g(x)$. Then the following statements hold:
\begin{enumerate}[$i)$]
 \item The skew polynomial $\phi(x^{\deg(h(x))}\varphi(h(x)))$ is a right divisor of $x^n-\lambda$.%
 \item The Hermitian dual   $C^{\perp_H}$ is  a
$\Theta$-$\lambda$-constacyclic code generated by $$\phi(x^{\deg(h(x))}\varphi(h(x))).$$
\end{enumerate}
\end{theorem}
\begin{proof}  From the proof of Theorem \ref{thm3.4},   we have
$$
 \varphi(g(x))(-\lambda x^{n-{\deg(h)}})x^{\deg(h)}\varphi(h(x))  =x^n-
{\lambda}. $$
  Then $ \phi(\varphi(g(x))(-\lambda x^{n-{\deg(h)}})) \phi(x^{\deg(h(x))}\varphi(h(x))) =\phi(x^n-
{\lambda})=x^n- {\lambda}. $ Therefore, $\phi(x^{\deg(h(x))}\varphi(h(x)))$ is a right divisor of $x^n-
{\lambda}$, which yields $i)$.

Since $g(x)h(x)=x^n-\lambda=0$ in $ \mathcal{R}[x;\Theta]/ \langle x^n-\lambda \rangle $, by   Lemma
\ref{lem2.10}, $$\langle \phi(x^{\deg(h(x))}\varphi(h(x)) )\rangle\subseteq C^{\perp_H}.$$ Since
$\phi(x^{\deg(h(x))}\varphi(h(x)))$ is a right divisor of $x^n-\lambda$,   by Proposition \ref{prop:free},
$$|\langle \phi(x^{\deg(h(x))}\varphi(h(x))) \rangle|=|\mathcal{R}|^{n-\deg(h(x))}=|C^{\perp_H}|.$$ Therefore,
$\langle \phi(x^{\deg(h(x))}\varphi(h(x))) \rangle= C^{\perp_H}$. This proves $ii)$.
\end{proof}

  Necessary and sufficient
conditions for a $\Theta$-$\lambda$-constacyclic code to be Hermitian self-dual are given. The proof follows as
an application of the proof of Theorem \ref{thm3.5}.
\begin{theorem}\label{thm3.7}
Assume that the order of $\Theta$ is $2$, $\lambda^2=1$ and $n$ is even, denoted by $n=2k$. Let $g(x)=
\sum_{i=0}^{k-1}g_ix^i+x^{k}$ be a right divisor of $x^n-\lambda$. Then the $\Theta$-$\lambda$-constacyclic code
generated by $g(x)$ is Hermitian self-dual if and only if
\begin{align*} (\sum_{i=0}^{k-1}g_ix^i+x^{k})(\Theta^{-k-1}(g_0^{-1})+\sum_{i=1}^{k-1}\Theta^{i-k-1}(g_0^{-1}g_{k-i})x^i+x^{k})=x^{n}-\lambda.
\end{align*}
\end{theorem}
\begin{remark} Suppose there is a  Hermitian self-dual $\Theta$-$\lambda$-constacyclic code. Then, by   Theorem
\ref{thm3.7}, we have $-\lambda=g_0\Theta^{-k-1}(g_0^{-1})$. Since $\lambda$ is fixed by $\Theta$, it follows
that $\lambda=-\Theta^{k+1}(g_0)g_0^{-1}$. As the order of $\Theta$ is $2$,
\begin{align*}
\lambda=\begin{cases} -1~~~~~~~~&\text{ if } k \text{ is odd,}\\
                       -\Theta(g_0)g_0^{-1}&\text{ if } k \text{ is even.}
\end{cases}
\end{align*}
Therefore,  if $k$ is odd and $\lambda\neq -1$, then there are no  Hermitian self-dual
$\Theta$-$\lambda$-constacyclic codes of length $2k$.
\end{remark}



%
%

\section{Skew Constacyclic Codes over
$\mathbb{F}_{p^m}+u\mathbb{F}_{p^m}$}

The class of finite chain rings of the form $\mathbb{F}_{p^m}+u\mathbb{F}_{p^m}$ has  widely been used as
alphabet in certain constacyclic codes (see, for example,  \cite{AmNe2008}, \cite{BoUd1999}, \cite{Di2009},
\cite{Di2010}, \cite{QiZa2006} and \cite{UdBo1999}). In this section, we characterize the structure of all
$\Theta$-$\lambda$-constacyclic codes over this ring  under the conditions where $\lambda $ is a unit in
$\mathbb{F}_{p^m}+u\mathbb{F}_{p^m}$ fixed by a given automorphism $\Theta$  and the length $n$ of codes is a
multiple of the order of $\Theta$. Moreover, the structures of Euclidean and Hermitian dual codes of skew cyclic
and skew negacyclic codes over this ring are determined as well.

 Recall that $\mathbb{F}_{p^m}+u\mathbb{F}_{p^m}$ is a
finite chain ring of nilpotency index $2$ and characteristic~$p$. Its only  maximal ideal is
$u\mathbb{F}_{p^m}$. The residue field $\mathcal{K}$ of $\mathbb{F}_{p^m}+u\mathbb{F}_{p^m}$ will be viewed as
the subfield $\mathbb{F}_{p^m}$ of $\mathbb{F}_{p^m}+u\mathbb{F}_{p^m}$.  Every automorphism
 of
$\mathbb{F}_{p^m}+u\mathbb{F}_{p^m} $ is of the form $\Theta_{\theta,\beta}(a+bu)=\theta(a)+\beta\theta(b) u$,
where $\theta\in \Aut(\mathbb{F}_{p^m})$ and $\beta\in \mathbb{F}_{p^m}^*$ (cf. Corollary \ref{cor-Auto} or
{\cite[Proposition 1]{Al-1991}}). For simplicity, where no confusion arises, the subscripts $ \theta$ and
$\beta$ will be dropped.

As the residue field $\mathcal{K}$ of $\mathbb{F}_{p^m}+u\mathbb{F}_{p^m}$ is viewed as the subfield
$\mathbb{F}_{p^m}$, the ring epimorphism   $ \bar{}: \mathbb{F}_{p^m}+u\mathbb{F}_{p^m} \rightarrow
\mathbb{F}_{p^m}$ can be viewed as the   reduction modulo $u$. For $f(x)\in
(\mathbb{F}_{p^m}+u\mathbb{F}_{p^m})[x;\Theta]$, $\overline{f(x)}$ denotes the isomorphic image in $
 \mathbb{F}_{p^m} [x;\theta]\subsetneq (\mathbb{F}_{p^m}+u\mathbb{F}_{p^m})[x;\Theta]$
of the componentwise reduction modulo $u$ of
 $f(x)$. Since  every skew polynomial in $ (\mathbb{F}_{p^m}+u\mathbb{F}_{p^m})[x;\Theta]$
 is viewed as  $  f_0(x)+uf_1(x)$, where $f_0(x), f_1(x)\in \mathbb{F}_{p^m}
 [x;\theta]$, we have
  $\overline{f_0(x)+uf_1(x)}=f_0(x)\in \mathbb{F}_{p^m}
 [x;\theta]$.


 For    $f(x) $ in  $ (\mathbb{F}_{p^m}+u\mathbb{F}_{p^m})[x;\Theta]$,
the multiplication rule allows the shifting of $u$ and powers of $x$ from the left to the right of $f(x)$ (and
vice versa) by changing the coefficients of $f(x)$. Then, for $\Omega\in\{u,x^i\mid i\in \mathbb{N}\}$, it is
meaningful to give
  the following notations:
\begin{enumerate}[$i)$]
\item ${\overleftarrow{f(x)}}^\Omega$  denotes the    skew polynomial such that
$f(x)\Omega=\Omega {\overleftarrow{f(x)}}^\Omega,$%
\item ${\overrightarrow{f(x)}}^\Omega$   denotes the skew polynomial
 such that   $\Omega f(x)= {\overrightarrow{f(x)}}^\Omega \Omega.$%
\end{enumerate}

\subsection{Classification of Skew Constacyclic Codes over
$\mathbb{F}_{p^m}+u\mathbb{F}_{p^m}$}

In this subsection, the  classification of $\Theta$-$\lambda$-constacyclic codes  is given in terms of
generators of left ideals in $(\mathbb{F}_{p^m}+u\mathbb{F}_{p^m})[x;\Theta]/\langle x^n-\lambda \rangle $.
These generators are uniquely determined under some   conditions. Their properties are also given.

 Let $C$ be a non-zero left ideal in $(\mathbb{F}_{p^m}+u\mathbb{F}_{p^m})[x;\Theta]/\langle  x^n-\lambda\rangle
  $ and let $A$ denote the set of
  all non-zero skew
polynomials of minimal degree in~$C$. Clearly, $A$ is non-empty. We consider three  cases: when there is a monic
skew polynomial in $A$, when there are no monic skew polynomials in $C$, and when there are no monic skew
polynomials in $A$ but  there is a monic skew polynomial in $C$.

\begin{theorem}\label{prop4.1}Let $C$ and $A$ be as above. Then:
\begin{enumerate}[$i)$]
\item If there exists a monic skew polynomial in $A$, then it is unique in $A$. In this case,
 $C=\langle g(x)\rangle$, where $g(x)$ is  the unique such skew polynomial.%
\item If there are no monic skew polynomials in $C$,
 then  there exists a unique skew polynomial $ g(x)=ug_1(x)$ in $A$  with  leading coefficient $u$. In
 this case, $C=\langle  g(x) \rangle$.   %
\item If there are no monic skew polynomials in $A$ but  there exists a monic skew polynomial in $C$, then there
exist a   unique skew polynomial $ g(x)=ug_1(x)$ in $A$ with  leading coefficient $u$ and a unique monic skew
polynomial $f(x)=f_0(x)+uf_1(x)$ of minimal degree in $C$ such that $\deg(f_1(x))<\deg(g_1(x))$. In this
case, $C=\langle g(x), f(x) \rangle$.   %
\end{enumerate}
\end{theorem}
\begin{proof}
To prove $i)$, assume that $g(x)$ and $g^\prime(x)$ are monic skew polynomials in $A$. Then  the degree of
$g(x)-g^\prime(x)$ is less than the degree of $g(x)$. By the minimality of   $\deg(g(x))$, $g(x)-g^\prime(x)=0$.
Hence, $g(x)$ is the unique monic skew polynomial in~$A$.

Let $c(x)\in C$. Then by the Right Division Algorithm, there exist unique skew polynomials $q(x)$ and $r(x)$ in
$(\mathbb{F}_{p^m}+u\mathbb{F}_{p^m})[x;\Theta]$ such that
\[c(x)=q(x)g(x)+r(x),\]
and $r(x)=0$ or $\deg(r(x))<\deg(g(x))$. Then
\[r(x)=c(x)-q(x)g(x)\in C.\] By  the  minimality of $\deg(g(x))$,
$r(x)=0$. Hence $c(x)=q(x)g(x)$, i.e., $C=\langle g(x)\rangle$.

 To prove $ii)$, assume there are no monic skew polynomials in $C$.  Without loss of generality, let $g(x)$
 be   a skew polynomial  in $A$ with   leading coefficient $u$.
 First, we show that $g(x)$ is a right multiple of $u$. Suppose that $g(x)$ has a unit coefficient  $a_i$ for some
$i<\deg(g(x))$. Then $u g(x) \in C$ is a non-zero skew polynomial having degree less than $\deg(g(x))$, which
contradicts the minimality of $\deg(g(x))$. Hence $g(x)$ is a right multiple of $u$, and we write $g(x)=u
g_1(x)$, where $g_1(x)$ is a monic skew polynomial in $\mathbb {F}_{p^m}[x;\theta]$.

For the uniqueness, suppose that $g^\prime(x)$ is
 a skew polynomial  in $A$ with   leading coefficient $u$. Then the degree of
$g(x)-g^\prime(x)$ is less than the degree of $g(x)$. By the minimality of $\deg(g(x))$, $g(x)-g^\prime(x)=0$.
Hence, $g(x)=u g_1(x)$ is the unique skew polynomial  in $A$ with   leading coefficient $u$.

 Now, we show that $C$ is generated by $g(x)=u g_1(x)$. Suppose
 that there exists $h(x)$ in $C$
 of minimal degree $\ell$ which is not a left multiple of
 $g(x)=u g_1(x)$. Moreover, $h(x)$ can be chosen to have   leading coefficient $u$.
   Then
 \begin{align*}
k(x):&=h(x)- u x^{\ell-\deg(g(x))}g_1(x)\\
     &=h(x)-{\overrightarrow{ x^{\ell-\deg(g(x))}}}^u ug_1(x)\\
     &=h(x)-    {\overrightarrow{ x^{\ell-\deg(g(x))}}}^u  {g}(x)\in C.
 \end{align*}
  If
$k(x)=0$, then $h(x)= {\overrightarrow{ x^{\ell-\deg(g(x))}}}^u {g}(x)$ which contradicts the assumption.
Suppose $k(x)\neq 0$. Then the degree of $k(x)$ is less than $\ell$ and $k(x)$ is not a left multiple of $g(x)$,
contradicting the choice of $h(x)$.

Finally, we prove $iii)$. Assume there are no monic skew polynomials in $A$ but there exists a monic skew
polynomial in $C$. It can be shown as in   $ii)$ that there is   a unique skew polynomial $g(x)=u g_1(x)$ in $
A$  with leading coefficient $u$.

Let $F(x)$ be a monic skew polynomial of minimal degree   in $C$. We view $F(x)=F_0(x)+u F_1(x)$, where
$F_0(x),F_1(x)\in \mathbb{F}_{p^m}[x;\theta]$.
 By the Right
Division Algorithm, there exist unique skew  polynomials $q(x)$ and $r(x)$ in $\mathbb {F}_{p^m}[x;\theta]$ such
that
\[F_1(x)=q(x)g_1(x)+r(x),\]
and $r(x)=0$ or $\deg(r(x))<\deg(g_1(x))$. Thus
\begin{align*}
  F(x)&=F_0(x)+u F_1(x)=F_0(x)+u q(x)g_1(x)+u r(x).
\end{align*}
We choose $f(x)=F(x)-u q(x)g_1(x)$, $f_0(x)=F_0(x)$ and $f_1(x)=r(x)$. Then   $f(x)=f_0(x)+uf_1(x)$ is a monic
skew polynomial of minimal degree in $C$ such that $\deg(f_1(x))<\deg(g_1(x))$.

The uniqueness of $ug_1(x)$ can be shown as in the proof of $ii)$. Suppose $t_0(x)+ut_1(x)$ is  a monic skew
polynomial of minimal degree in $C$ such that $\deg(t_1(x))<\deg(g_1(x))$. Then $\langle uf_0(x)\rangle =uC=
\langle ut_0(x)\rangle $. Hence, by the proof of $ii)$, $ f_0(x)=t_0(x)$. Note that
$u(f_1(x)-t_1(x))=(f_0(x)+uf_1(x))-(t_0(x)+ut_1(x))\in C$. Then $u(f_1(x)-t_1(x))$ is the zero  or
$\deg(f_1(x)-t_1(x))\leq\max\{\deg(f_1(x)),\deg(t_1(x))\}$. If the later case occurs, then
$\deg(f_1(x)-t_1(x))<\deg(g_1(x))$, which contradicts the minimality of $\deg(g_1(x))$. Hence $f_1(x)-t_1(x)=0$.

Let   $B$ be the set of all non-zero skew polynomials in $C$ with degree less than $\deg(f(x))$. Then the
leading coefficients of all skew polynomials in $B$ are multiple of $u$. Since $ug_1\in A$, we have
$\deg(ug_1(x))<\deg(f(x))$, and hence $u g_1(x)\in B$. Using arguments similar to the third paragraph in the
proof of $ii)$, $B$ is contained in the left ideal generated by $u g_1(x)$.

To show that $C$ is generated by $\{g(x)=ug_1(x),f(x)=g_0(x)+ug_1(x)\}$, let \mbox{$c(x)\in C$.} Then there
exist unique skew polynomials $q^\prime(x)$ and $r^\prime(x)$ in
$(\mathbb{F}_{p^m}+u\mathbb{F}_{p^m})[x;\Theta]$ such that
\[c(x)=q^\prime(x)f(x)+r^\prime(x),\]
and $r^\prime(x)=0$ or $\deg(r^\prime(x))<\deg(f(x))$.  If $r^\prime(x)=0$, we are done.   Assume that
$\deg(r^\prime(x))<\deg(f(x))$. Then $r^\prime(x)\in B$ and hence $r^\prime(x)=m(x)g(x)$ for some $m(x)\in
(\mathbb{F}_{p^m}+u\mathbb{F}_{p^m})[x;\Theta]$. Hence
\[c(x)=q^\prime(x)f(x)+r^\prime(x)=q^\prime(x)f(x)+m(x)g(x).\]
Therefore, $C$ is   generated by $\{g(x)=u
 g_1(x), f(x)=f_0(x)+u f_1(x) \}$.
\end{proof}

For convenience, we split the left ideals of $(\mathbb{F}_{p^m}+u\mathbb{F}_{p^m})[x;\Theta]/\langle
x^n-\lambda\rangle
  $ into three  types: Type  LI-$1$  refers to the zero ideal or a left ideal satisfying Theorem \ref{prop4.1} $i)$,
     type   LI-$2$ refers to a left ideal  satisfying Theorem \ref{prop4.1} $ii)$,
     and   type   LI-$3$ refers to a left ideal  satisfying Theorem \ref{prop4.1}
     $iii)$.

More properties of left ideals of each type are given
     in the following propositions.

\begin{proposition}
A left ideal of type {\rm LI}-$1$  is   principal and generated by a monic right divisor $g(x)$ of $x^n-\lambda$
   in $(\mathbb{F}_{p^m}+u\mathbb{F}_{p^m})[x;\Theta] $.
   Moreover, if we view $g(x)=g_0(x)+u g_1(x)$, where $ g_0(x),g_1(x) \in
   \mathbb{F}_{p^m}[x;\theta]$, then $\deg (g_1(x))< \deg(g_0(x)) $ and $g_0(x) $ is a monic right divisor of
  $ { x^n-\overline\lambda} $ in $\mathbb {F}_{p^m}[x;\theta] $. %
\end{proposition}
\begin{proof} Let $C$ be a left ideal of type {\rm LI}-$1$.
If $C=\{0\}$, then $C=\langle 0 \rangle = \langle x^n-\lambda
  \rangle $ has  the desired properties.

  Suppose $C$ is non-zero. We prove that the generator polynomial
  $g(x)$
  in Theorem \ref{prop4.1} $i)$ satisfies these properties. Recall that $ g(x)$
  is the unique monic skew polynomial in $A$,
   the set of all non-zero skew
polynomials of minimal degree in $C$

First, we show that $g(x)$ is a right divisor of $x^n-\lambda$ in
$(\mathbb{F}_{p^m}+u\mathbb{F}_{p^m})[x;\Theta]$. By the Right Division Algorithm, there exist unique skew
polynomials $q (x)$ and $r (x)$ in $(\mathbb{F}_{p^m}+u\mathbb{F}_{p^m})[x;\Theta]$ such that
\[x^n-\lambda=q (x)g(x)+r (x),\]
and $r (x)=0$ or $\deg(r (x))<\deg(g(x))$. Then
\[r (x)=-q (x)g(x)+(x^n-\lambda)\in C.\] By the  minimality of $\deg(g(x))$,
$r (x)=0$. Hence $g(x)$ is a right divisor of $x^n-\lambda$.

Finally, we write $g(x)=g_0(x)+u g_1(x)$, where $ g_0(x),g_1(x) \in \mathbb{F}_{p^m}[x;\theta]$. Since $g(x)$ is
monic, it is clear that $g_0(x)$ is monic and $\deg(g_1(x))<\deg(g(x))=\deg(g_0(x))$. As $g(x)$ is a right
divisor of $x^n-\lambda$ in $(\mathbb{F}_{p^m}+u\mathbb{F}_{p^m})[x;\Theta]$, there exists $p(x)$ in
$(\mathbb{F}_{p^m}+u\mathbb{F}_{p^m})[x;\Theta]$ such that
\[x^n-\lambda=p(x)(g_0(x)+u g_1(x)).\] Computing modulo $u$, we have
$x^n-\overline{\lambda}=\overline{p(x)}g_0(x)$ in $\mathbb {F}_{p^m}[x;\theta]$. This means $g_0(x)$ is a monic
right divisor of $x^n-\overline{\lambda}$ in $\mathbb {F}_{p^m}[x;\theta]$.
\end{proof}

\begin{proposition}\label{prop-li2}
A left ideal of type {\rm LI}-$2$  is   principal and generated by $g(x)=u g_1(x)$,
     where $g_1(x)$ is a monic right divisor
      of $x^n-\overline{\lambda}$ in $ \mathbb {F}_{p^m}  [x;\theta] $ such that $\deg(g_1(x))<n$.
\end{proposition}
\begin{proof} Let $C$ be a left ideal of type {\rm LI}-$2$.
We prove that the generator polynomial
  $g(x)=ug_1(x)$
  in Theorem \ref{prop4.1} $ii)$ satisfies the desired properties.
Recall that $g(x)=u g_1(x)$ is the unique skew polynomial with leading coefficient $u$ in $ A$, the set of all
non-zero skew polynomials of minimal degree in $C$. Clearly, $\deg(g_1(x))<n$.    By the Right Division
Algorithm, there exist unique skew polynomials $q(x)$ and $r(x)$ in $\mathbb {F}_{p^m}[x;\theta]$ such that
\[x^n-\overline{\lambda}=q(x)g_1(x)+r(x),\]
and $r(x)=0$ or $\deg(r(x))<\deg(g_1(x))$. Since $u (x^n-\bar{\lambda})=u (x^n-{\lambda})$, we have
\begin{align*}
u r(x)&=-u q(x)g_1(x)+u (x^n-\bar{\lambda})\\
     &=-{\overrightarrow{q(x)}}^u u g_1(x)+u (x^n-{\lambda})\\
     &=-{\overrightarrow{q(x)}}^u g(x)+u (x^n-{\lambda})\in C .
\end{align*}
By  the  minimality of $\deg(g(x))$, $u r(x)=0$. As $r(x)\in \mathbb {F}_{p^m}[x;\theta]$, $r(x)=0$. Hence
$g_1(x)$ is a right divisor of $x^n-\overline{\lambda}$ in $\mathbb {F}_{p^m}[x;\theta]$.
\end{proof}
\begin{proposition}\label{prop-mx}
A left ideal of type {\rm LI}-$3$  is   generated by $\{g(x)=u g_1(x), f(x)=f_0(x)+u f_1(x) \}$,
   where $f_0(x),$ $f_1(x),$ $g_1(x)\in \mathbb {F}_{p^m}  [x;\theta] $ satisfy the following properties:
  \begin{enumerate}[ $i)$]
    \item $g_1(x),f_0(x)$ are monic,
    \item $\deg(f_1(x))<\deg(g_1(x)) {<\deg(f_0(x))<n}$,
    \item $g_1(x)$ is a right divisor of $f_0(x)  $ in $ \mathbb {F}_{p^m}  [x;\theta]
    $,
    \item $f_0(x) $ is a right divisor of ${ x^n-\overline\lambda}$ in $ \mathbb {F}_{p^m}  [x;\theta]
    $. %
  \end{enumerate}
  Moreover,   if $\lambda\in \mathbb {F}_{p^m}$, then $g_1(x)$ is a right divisor of $  {\overleftarrow{ \left(\displaystyle\frac{ {x^n- \lambda} }{f_0(x)}\right) }}^u  f_1(x)$ in $ \mathbb {F}_{p^m}
  [x;\theta].
    $
\end{proposition}
\begin{proof}
Let $C$ be a left ideal of type {\rm LI}-$3$. We prove that the generator set
  $\{g(x)=u g_1(x), f(x)=f_0(x)+u f_1(x) \}$
  in Theorem \ref{prop4.1} $iii)$ satisfies the desired properties.
Recall that   $g(x)=u g_1(x)$ is the unique skew polynomial with the leading coefficient $u$ in $ A$, the set of
all non-zero skew polynomials of minimal degree in $C$ and $f_0(x)+uf_1(x)$ is the unique monic skew polynomial
of minimal degree in $C$ such that $\deg(f_1(x))<\deg(g_1(x))$.

 Properties $i)$ and $ii)$ are clear.  Property $iii)$ can be proved by a similar
 argument in the case for Proposition \ref{prop-li2} with $x^n-\bar{\lambda}$ replaced by $f_0(x)$.

Note that $u f_0(x)$ is a skew  polynomial of minimal degree in $\langle u f_0(x) \rangle$. Using   arguments
similar  to the proof of Proposition \ref{prop-li2}, $f_0(x)$ is a right divisor of $x^n-\overline{\lambda}$ in
$\mathbb {F}_{p^m} [x;\theta]$. Hence, property $iv)$ is proved.

Finally, it is straightforward  to see that if $\lambda\in \mathbb{F}_{p^m}$, then $\bar{\lambda}=\lambda$. Thus
\begin{align*}
  \frac{ x^n- {\lambda} }{f_0(x)}(f_0(x)+u f_1(x))&=\frac{ x^n-  {\lambda} }{f_0(x)} u f_1(x)\\
                                                    &=u  {\overleftarrow {\left(\frac{ x^n- {\lambda} }{f_0(x)}\right)}}^u  f_1(x) %
\\
                                                    &\in C\cap u ((\mathbb{F}_{p^m}+u\mathbb{F}_{p^m})[x;\Theta]/\langle
                                                    x^n-\lambda\rangle).
\end{align*}
Note that $C\cap u ((\mathbb{F}_{p^m}+u\mathbb{F}_{p^m})[x;\Theta]/\langle x^n-\lambda\rangle)$ is a left ideal
in $ (\mathbb{F}_{p^m}+u\mathbb{F}_{p^m})[x;\Theta]/$ $\langle x^n-\lambda\rangle$   containing $g(x)=u g_1(x)$
as a skew polynomial of minimal degree. Since $C\cap u( (\mathbb{F}_{p^m}+u\mathbb{F}_{p^m})[x;\Theta]/\langle
x^n-\lambda\rangle)$ does not contain any monic element, by Proposition \ref{prop-li2}, it is generated by
$g(x)=u g_1(x)$. Hence $g_1(x)$ is a right divisor of $ {\overleftarrow{ \left(\displaystyle\frac{ x^n-
{\lambda} }{f_0(x)}\right)}}^u f_1(x)$.
\end{proof}

\begin{example}\label{ex4.2}Figures
\ref{lattice-identity} and  \ref{lattice-non-identity} show
  the ideal lattices  of
$(\mathbb{F}_{3}+u\mathbb{F}_{3})[x]/\langle x^2-1\rangle$ and $(\mathbb{F}_{3}+u\mathbb{F}_{3})[x;\Theta_{{\rm
id},2}]/\langle x^2-1\rangle$, where $\Theta_{{\rm id},2}(a+bu)=a+2bu$ for all $a,b\in \mathbb{F}_{3}$.
 The subscripts $1$, $2$ and $3$ indicate
types LI-$1$, LI-$2$ and LI-$3$, respectively.

\begin{figure}[h!bt]
\centering
\begin{small}
\begin{tikzpicture}[description/.style={fill=white,inner sep=0pt}]
    \matrix (m) [matrix of math nodes, row sep=1.5em, column sep=0em]
    {       &       &       & \langle 1  \rangle_{ 1}   &       &       &  \\
            &       & \langle u,x+1 \rangle_{ 3} &       &\langle u,x+2 \rangle_{ 3}  &       &  \\
            &       &       &\langle u \rangle_{ 2}     &       &       &  \\
            &       & \langle x+1 \rangle_{ 1}  &       &\langle x+2 \rangle_{ 1}   &   &  \\
            &       &\langle u(x+1) \rangle_{ 2} &       &\langle u(x+2 ) \rangle_{ 2} &       &  \\
             &      &      &\langle  0    \rangle_{ 1} &       &       &  \\
        };
    \path[-]
        (m-1-4) edge node[auto] {} (m-2-3)
        (m-1-4) edge node[auto] {} (m-2-5)

        (m-2-3) edge node[auto] {} (m-4-3)
        (m-2-3) edge node[auto] {} (m-3-4)
        (m-2-5) edge node[auto] {} (m-4-5)

        (m-2-5) edge node[auto] {} (m-3-4)
        (m-4-5) edge node[auto] {} (m-5-5)

        (m-3-4) edge node[auto] {} (m-5-5)

        (m-4-3) edge node[auto] {} (m-5-3)
        (m-3-4) edge node[auto] {} (m-5-3)
        (m-5-3) edge node[auto] {} (m-6-4)
        (m-5-5) edge node[auto] {} (m-6-4)
        ;
\end{tikzpicture}
\end{small}
  \caption{The ideal lattice of $(\mathbb{F}_{3}+u\mathbb{F}_{3})[x]/\langle x^2-1\rangle$}\label{lattice-identity}
\end{figure}

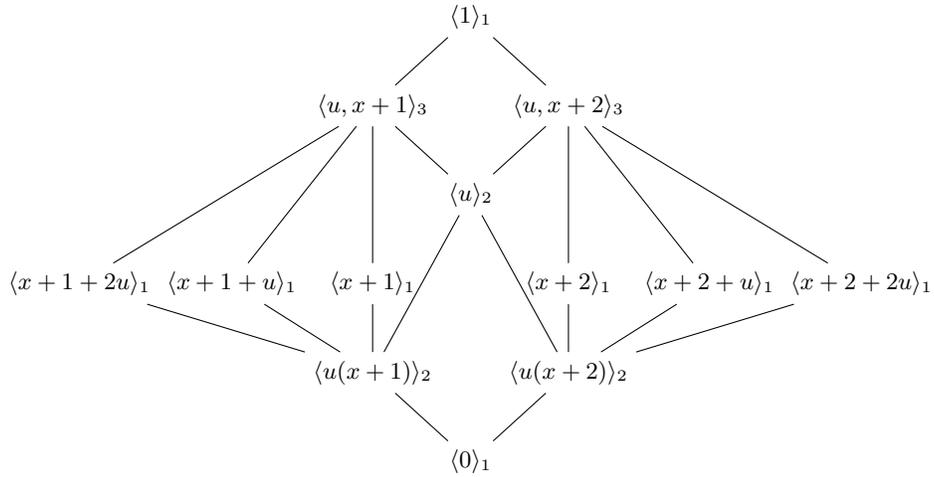
\begin{figure}[h!bt]
\centering
 \begin{small}
\begin{tikzpicture}[description/.style={fill=white,inner sep=0pt}]
    \matrix (m) [matrix of math nodes, row sep=2em, column sep=0em]
    {       &       &       & \langle 1  \rangle_{ 1}   &       &       &  \\
            &       & \langle u,x+1 \rangle_{ 3} &       &\langle u,x+2 \rangle_{ 3}  &       &  \\
            &       &       &\langle u \rangle_{ 2}    &       &       &  \\
     \langle x+1 +2u  \rangle_{ 1}  & \langle x+1+u \rangle_{ 1}    & \langle x+1 \rangle_{ 1}   &       &\langle x+2 \rangle _{ 1}   &\langle x+2+u  \rangle _{ 1}  &\langle x+2+2u\rangle _{ 1}  \\
            &       &\langle u(x+1) \rangle_{ 2}  &       &\langle u(x+2 ) \rangle_{ 2} &       &  \\
             &      &      &\langle  0    \rangle_{ 1}  &       &       &  \\
        };
    \path[-]
        (m-1-4) edge node[auto] {} (m-2-3)
        (m-1-4) edge node[auto] {} (m-2-5)
        (m-2-3) edge node[auto] {} (m-4-1)
        (m-2-3) edge node[auto] {} (m-4-2)
        (m-2-3) edge node[auto] {} (m-4-3)
        (m-2-3) edge node[auto] {} (m-3-4)
        (m-2-5) edge node[auto] {} (m-4-5)
        (m-2-5) edge node[auto] {} (m-4-6)
        (m-2-5) edge node[auto] {} (m-4-7)
        (m-2-5) edge node[auto] {} (m-3-4)
        (m-4-5) edge node[auto] {} (m-5-5)
        (m-4-6) edge node[auto] {} (m-5-5)
        (m-4-7) edge node[auto] {} (m-5-5)
        (m-3-4) edge node[auto] {} (m-5-5)
        (m-4-1) edge node[auto] {} (m-5-3)
        (m-4-2) edge node[auto] {} (m-5-3)
        (m-4-3) edge node[auto] {} (m-5-3)
        (m-3-4) edge node[auto] {} (m-5-3)
        (m-5-3) edge node[auto] {} (m-6-4)
        (m-5-5) edge node[auto] {} (m-6-4)
        ;
\end{tikzpicture}
\end{small}
  \caption{The ideal lattice of $(\mathbb{F}_{3}+u\mathbb{F}_{3})[x;\Theta_{{\rm id},2}]/\langle x^2-1\rangle$}\label{lattice-non-identity}
\end{figure}
Note that Figure \ref{lattice-identity} is embedded in Figure \ref{lattice-non-identity}.
\end{example}

\subsection{Euclidean  Dual Codes of Skew Cyclic and Skew Negacyclic Codes over
$\mathbb{F}_{p^m}+u\mathbb{F}_{p^m}$}

We study  the structures of the Euclidean dual codes of skew cyclic and skew negacyclic codes over
$\mathbb{F}_{p^m}+u\mathbb{F}_{p^m}$.  For this purpose, we assume that $\lambda=\pm 1$. Since
$\bar{\lambda}=\lambda\in \mathbb{F}_{p^m}$  is always fixed by any automorphism,  $\Theta$ can be arbitrary.
However, the length $n$ of codes is assumed to be a multiple of the order of $\Theta$.

As $\lambda^2=1$, by Lemma \ref{lem2.pre9},   the Euclidean dual codes of skew cyclic and skew negacyclic codes
are again skew cyclic and skew negacyclic, respectively.  Their generators are given through the unique
representation of the original codes and the ring anti-monomorphism $\varphi$ defined in Proposition
\ref{prop2.7}, where $
 \varphi(\sum_{ i=0}^t a_ix^{i} )= \sum_{
i=0}^t x^{-i}a_i$.

\begin{theorem}\label{thm4.5}
Let  $\lambda \in \{-1,1\}$. Then the Euclidean dual code of a left ideal in
$(\mathbb{F}_{p^m}+u\mathbb{F}_{p^m})[x;\Theta]/\langle x^n-\lambda\rangle $  is also a left ideal
in$(\mathbb{F}_{p^m}+u\mathbb{F}_{p^m})[x;\Theta]/\langle x^n-\lambda\rangle $ determined as follows:
\begin{enumerate}[\textnormal{LI}-$1^\perp.$]
  \item If  $C=\langle  g_0(x)+ug_1(x)\rangle $, then $C^\perp=\langle  x^{n-\deg(g_0(x))}\varphi\left(
  \displaystyle\frac{x^n-\lambda}{ g_0(x)+ug_1(x)}\right)\rangle
  $.
   \item If $C=\langle  ug_1(x) \rangle$, then   $C^\perp=\langle  u ,x^{n-\deg( g_1(x))}\varphi\left(
     \displaystyle\frac{x^n-\lambda}{ g_ 1(x)} \right)
     \rangle
  .$ %
  \item If $C=\langle  ug_1(x), f_0(x)+uf_1(x) \rangle$, then there exists $m(x)\in \mathbb{F}_{p^m}[x;\theta] $ such that $m(x)g_1(x)= {\overleftarrow{\left(\displaystyle\frac{x^n-\lambda}{ f_0(x) }\right)}}^u f_1(x)$ and $$C^\perp=\langle  x^{n-\deg( f_0(x))}\varphi\left(
   \displaystyle\frac{x^n-\lambda}{ f_0(x)}u\right), x^{n-\deg( g_1(x))}\varphi\left(
     \displaystyle\frac{x^n-\lambda}{ g_ 1(x)}-um(x)\right)\rangle
  .$$

\end{enumerate}
 \end{theorem}

 For    LI-$1^\perp$, the
Euclidean dual code of   type LI-$1$ code is determined in Theorem~\ref{thm3.4}  and it is shown to be type
LI-$1$.
  Moreover,
$(C^\perp)^\perp=C$ implies that $C$ is type LI-$1$ if and only if $C^\perp$ is type LI-$1$. However, this does
not need to be true for types LI-$2$ and LI-$3$ (see Example \ref{ex-last}).



In  LI-$2^\perp$ and LI-$3^\perp$, $f_0(x)$, $g_1(x)$ are right divisors of $x^n-\lambda$ in
$\mathbb{F}_{p^m}[x;\theta]$. Since $x^n-\lambda$ is central, it follows from  (\ref{eq:frac-skew}) that
\begin{align}
f_0(x)\displaystyle\frac{x^n-\lambda}{f_0(x)}&=x^n-\lambda=\displaystyle\frac{x^n-\lambda}{f_0(x)}f_0(x)
,\label{eq:4.2-frac-xn-a-f0}\\
g_1(x)\displaystyle\frac{x^n-\lambda}{g_1(x)}&=x^n-\lambda=\displaystyle\frac{x^n-\lambda}{g_1(x)}g_1(x)
.\label{eq:4.2-frac-xn-a-g1}
\end{align}
These two facts  and the centrality of $x^n-\lambda$ will be frequently used in the following proofs.

\noindent{\bf Proof of LI-$2^\perp$.} Let $D:=\langle  u ,x^{n-\deg( g_1(x))}\varphi\left(
     \displaystyle\frac{x^n-\lambda}{ g_ 1(x)} \right)
     \rangle$. Clearly,  $u\in C^\perp$. From (\ref{eq:4.2-frac-xn-a-g1}), it follows  that  $
    (ug_1(x))\displaystyle\frac{x^n-\lambda}{ g_1(x)}=u(x^n-\lambda)=0$ in $ (\mathbb{F}_{p^m}+u\mathbb{F}_{p^m})[x;\Theta]/\langle
    x^n-\lambda\rangle.$ Hence $D\subseteq
C^\perp$ is concluded via    Lemma \ref{lem2.9}.

In the other direction, we note that $C^\perp$ is of either type LI-$2 $ or LI-$3 $. If $C^\perp=\langle us_1(x)
\rangle$ is of type  LI-$2$, then $C^\perp\subseteq \langle u \rangle \subseteq D$. Suppose that
$C^\perp:=\langle us_1(x),t_0(x)+ut_1(x)\rangle$ is of type LI-$3 $.  Clearly, $us_1(x),ut_1(x)\in \langle u
\rangle \subseteq D$.

Since $ ug_1(x)\in C$ and $t_0(x)+ut_1(x)\in C^\perp$, it follows from Lemma~\ref{lem2.9}
\begin{align*}
  0&=(ug_1(x))\varphi^{-1}(x^{-\deg(t_0(x))}(t_0(x)+ut_1(x)))\\
   &=ug_1(x)\varphi^{-1}(x^{- \deg(t_0(x))}t_0(x))
\end{align*}
in $ ( \mathbb{F}_{p^m}+u\mathbb{F}_{p^m}) [x;\Theta]/\langle x^n-\lambda \rangle$.
 Thus $  g_1(x)\varphi^{-1}(x^{-
\deg(t_0(x))}t_0(x)) =0$. Hence, in~$( \mathbb{F}_{p^m}+u\mathbb{F}_{p^m}) [x;\Theta]$,
\begin{align}
g_1(x)\varphi^{-1}(x^{
                             -\deg(t_0(x))}t_0(x)) =l_1(x)(x^n-\lambda)  =(x^n-\lambda) l_1(x),\label{eq:g1-LI2}
\end{align}
{for some } $l_1(x)\in
                             \mathbb{F}_{p^m}[x;\theta].$
Note that
\begin{align}\label{eq:deg-t0-LI2}
\deg(t_0(x))=\deg(l_1(x))+n-\deg(g_1(x)).
\end{align}
With the notation in (\ref{eq:4.2-frac-xn-a-g1}), left cancellation of (\ref{eq:g1-LI2}) by $g_1(x)$ gives
                              \[
\displaystyle\frac{x^n-\lambda}{g_1(x)} l_1(x)
                              = \varphi^{-1}(x^{-
                             \deg(t_0(x))}t_0(x)),
\] and hence, by (\ref{eq:deg-t0-LI2}),
\begin{align*}
t_0(x)  &=x^{ \deg(t_0(x))}\varphi\left(\frac{x^n-\lambda}{g_1(x)} l_1(x)\right)\\
        &=x^{\deg(l_1(x))+n-\deg(g_1(x))}\varphi(
        l_1(x))\varphi\left(\frac{x^n-\lambda}{g_1(x)}\right)\\
        &=x^{\deg(l_1(x))} {\overrightarrow{\varphi(
        l_1(x))}}^{x^{n-\deg(g_1(x))}}   {x^{n-\deg(g_1(x))}}\varphi\left(\frac{x^n-\lambda}{g_1(x)}\right)\in D.
\end{align*}
Consequently, $t_0(x)+ut_1(x)\in D$. As desired, $C^\perp\subseteq D$. \hfill $\square$

\vskip1em \noindent{\bf Proof of LI-$3^\perp$.} Since $\lambda\in \mathbb{F}_{p^m}$, it follows from Proposition
\ref{prop-mx} that $g_1(x)$ is a right divisor of $ {\overleftarrow{ \left(\displaystyle\frac{ x^n- {\lambda}
}{f_0(x)}\right)}}^u f_1(x).$ Then there exists $m(x)\in \mathbb{F}_{p^m}[x;\theta] $ such that
\begin{align}\label{eq-m} m(x)g_1(x)= {\overleftarrow{\left(\displaystyle\frac{x^n-\lambda}{ f_0(x) }\right)}}^u
f_1(x).
\end{align}
   Let $D:=\langle
x^{n-\deg( f_0(x))}\varphi\left(
   \displaystyle\frac{x^n-\lambda}{ f_0(x)}u\right), x^{n-\deg( g_1(x))}\varphi\left(
     \displaystyle\frac{x^n-\lambda}{ g_
     1(x)}-um(x)\right)\rangle.$
It follows from (\ref{eq-m}) that
\begin{align}\label{eq:4.4}u m(x)g_1(x)=u
{\overleftarrow{\left(\displaystyle\frac{x^n-\lambda}{ f_0(x) }\right)}}^u
f_1(x)=\displaystyle\frac{x^n-\lambda}{ f_0(x)
     }uf_1(x).
\end{align}
Multiplying  on the left    of (\ref{eq:4.4}) by $f_0(x)$, we have
\begin{align*}
f_0(x)u m(x)g_1(x)  &= f_0(x)\displaystyle\frac{x^n-\lambda}{ f_0(x)
     }uf_1(x) \\
                    &= {(x^n-\lambda)} uf_1(x)   ~~~~~~~~(\text{using (\ref{eq:4.2-frac-xn-a-f0})})\\
                    &=uf_1(x){(x^n-\lambda)}\\
                    &=uf_1(x)\displaystyle\frac{x^n-\lambda}{g_1(x)}g_1(x)  ~~~(\text{using
                    (\ref{eq:4.2-frac-xn-a-g1})}).
\end{align*}
     Hence,
\begin{align}\label{eq:4.5}
 f_0(x)u m(x) =uf_1(x)\displaystyle\frac{x^n-\lambda}{g_1(x)},
\end{align}
and
\begin{align}\label{eq:deg-m}
 \deg(m(x))=n+\deg(f_1(x))-\deg(f_0(x))-\deg(g_1(x)).
\end{align}

Now, we observe the following:
\begin{enumerate}[$a)$]
\item Since $u^2=0$, we have
    \begin{align}\label{eq:ob1}
    ug_1(x)\frac{x^n-\lambda}{ f_0(x)}u=0.
    \end{align}%
\item Using $u^2=0$ and (\ref{eq:4.2-frac-xn-a-g1}), we conclude that
    \begin{align}\label{eq:ob2}
    ug_1(x)\left(\frac{x^n-\lambda}{ g_1(x)}-um(x)\right)&=ug_1(x)\frac{x^n-\lambda}{ g_1(x)}
    =u(x^n-\lambda).
    \end{align}%
\item It follows from  $u^2=0$ and (\ref{eq:4.2-frac-xn-a-f0}) that
    \begin{align} \label{eq:ob3}
        (f_0(x)+uf_1(x) )(\frac{x^n-\lambda}{ f_0(x)}u)&=f_0(x)\frac{x^n-\lambda}{ f_0(x)}u
        =(x^n-\lambda)u=u(x^n-\lambda).
    \end{align}
\item Since   $g_1(x)$ is a right divisor of $f_0(x)$, by (\ref{eq:frac-skewR}) and (\ref{eq:4.2-frac-xn-a-g1}),
we have
 \begin{align}\label{eq:ob-pre4}
    f_0(x)\frac{x^n-\lambda}{
    g_1(x)}=\left(\frac{f_0(x)}{g_1(x)}g_1(x)\right) \frac{x^n-\lambda}{
    g_1(x)}&=\frac{f_0(x)}{g_1(x)}\left(g_1(x)\frac{x^n-\lambda}{
    g_1(x)}\right)\notag\\
    &=\frac{f_0(x)}{g_1(x)}(x^n-\lambda).
    \end{align}%
The next equation follows from $u^2=0$,  (\ref{eq:4.5}) and  (\ref{eq:ob-pre4})   %
    \begin{align}\label{eq:ob4}
    (f_0(x)+uf_1(x) )\left(\frac{x^n-\lambda}{
    g_1(x)}-um(x)\right)
    &=f_0(x)\frac{x^n-\lambda}{
    g_1(x)}+uf_1(x)\frac{x^n-\lambda}{
    g_1(x)}\notag\\
    &~~~~~~~~~~~~~~~~~~~~\,-f_0(x)um(x)\notag\\
    &=\frac{f_0(x)}{g_1(x)}(x^n-\lambda).
    \end{align}%
\end{enumerate}
Equations (\ref{eq:ob1})-(\ref{eq:ob3}) and   (\ref{eq:ob4}) equal $0$ in
$(\mathbb{F}_{p^m}+u\mathbb{F}_{p^m})[x;\Theta]/\langle x^n-\lambda\rangle.$ Thus, by Lemma \ref{lem2.9},
$D\subseteq C^\perp$.


For the reverse inclusion, we note that $C^\perp$ is of type LI-$2 $ or LI-$3 $.  First, suppose that
$C^\perp:=\langle us_1(x) \rangle$ is of type LI-$2 $. Since $f_0(x)+uf_1(x)\in C$ and $us_1(x)\in C^\perp$, the
Euclidean orthogonality and Lemma~\ref{lem2.9} imply  that
\[ (f_0(x)+uf_1(x))\varphi^{-1}(x^{-\deg(s_1)}us_1(x))=0\]
 in $(\mathbb{F}_{p^m}+u\mathbb{F}_{p^m})[x;\Theta]/\langle
x^n-\lambda\rangle$. Hence, in $(\mathbb{F}_{p^m}+u\mathbb{F}_{p^m})[x;\Theta]$,
\begin{align}\label{eq:4.12}
f_0(x)\varphi^{-1}(x^{-\deg(s_1(x))}us_1(x))=ul(x)(x^n-\lambda)=(x^n-\lambda)ul(x),
\end{align} for some $l(x)\in \mathbb{F}_{p^m} [x;\theta]$.
 Moreover, $
\deg(s_1(x))=n+\deg(l(x))-\deg(f_0(x)).$ It follows from (\ref{eq:4.2-frac-xn-a-f0}) and (\ref{eq:4.12}) that
\begin{align*}
\varphi^{-1}(x^{-({n+\deg(l(x))-\deg(f_0(x))})} us_1(x))&=\varphi^{-1}(x^{-\deg(s_1(x))} us_1(x)) =
\displaystyle\frac{x^n-\lambda}{f_0(x)}ul(x).
\end{align*}
Since $\varphi$ is a ring anti-monomorphism, we conclude that
\begin{align*}
x^{-({n+\deg(l(x))-\deg(f_0(x))})} us_1(x)=\varphi\left(\frac{x^n-\lambda}{f_0(x)}ul(x)\right)
=\varphi(l(x))\varphi\left(\frac{x^n-\lambda}{f_0(x)}u\right) .
\end{align*}
Consequently,
\begin{align*}  us_1(x)&=
x^{n+\deg(l(x))-\deg(f_0(x))}
\varphi(l(x))\varphi\left(\frac{x^n-\lambda}{f_0(x)}u\right)\notag\\
&=x^{ \deg(l(x)) }{\overrightarrow{\varphi(l(x))}}^{x^{n -\deg(f_0(x))}}x^{n
-\deg(f_0(x))}\varphi\left(\frac{x^n-\lambda}{f_0(x)}u\right) \in D.
\end{align*}

Next, suppose that  $C^\perp:=\langle us_1(x),t_0(x)+ut_1(x)\rangle$  is of type LI-$3$.  Using arguments
similar  to those  above, $f_0(x)+uf_1(x)\in C$ and $us_1(x)\in C^\perp$ imply $ us_1(x) \in D.$

Since $ ug_1(x)\in C$ and $t_0(x)+ut_1(x)\in C^\perp$, it follows from Lemma~\ref{lem2.9} that
\begin{align*}
  0&=ug_1(x)\varphi^{-1}(x^{-\deg(t_0(x))}(t_0(x)+ut_1(x))) =ug_1(x)\varphi^{-1}(x^{- \deg(t_0(x))}t_0(x)),
\end{align*}
in $ ( \mathbb{F}_{p^m}+u\mathbb{F}_{p^m}) [x;\Theta]/\langle x^n-\lambda \rangle$. Thus $
g_1(x)\varphi^{-1}(x^{- \deg(t_0(x))}t_0(x)) =0,$  and hence, in $( \mathbb{F}_{p^m}+u\mathbb{F}_{p^m})
[x;\Theta]$,
\begin{align}
g_1(x)\varphi^{-1}(x^{
                             -\deg(t_0(x))}t_0(x))&=l_1(x)(x^n-\lambda)  =(x^n-\lambda) l_1(x), \label{eq-deg}
\end{align}
{for some } $l_1(x)\in
                             \mathbb{F}_{p^m}[x;\theta].$ Note
                             that
\begin{align}\label{eq:deq-t0}\deg( t_0(x))=n+\deg(l_1(x))-\deg(g_1(x)).\end{align}
In the notation of (\ref{eq:4.2-frac-xn-a-g1}),  the left cancellation of (\ref{eq-deg}) by  $g_1(x)$ implies
\begin{align}\label{eq:4.16}
\varphi^{-1}(x^{-
                             \deg(t_0(x))}t_0(x))=\displaystyle\frac{x^n-\lambda}{g_1(x)} l_1(x) ,
\end{align}
 and hence
\begin{align}
t_0(x)  &=x^{ \deg(t_0(x))}\varphi\left(\frac{x^n-\lambda}{g_1(x)} l_1(x)\right)=x^{ \deg(t_0(x))}\varphi(
        l_1(x))\varphi\left(\frac{x^n-\lambda}{g_1(x)}\right).\label{eq:t0}
\end{align}
By Lemma~\ref{lem2.9}, in $ ( \mathbb{F}_{p^m}+u\mathbb{F}_{p^m}) [x;\Theta]/\langle x^n-\lambda \rangle$,
\begin{align*}
  0&=(f_0(x)+uf_1(x))\varphi^{-1}(x^{-\deg(t_0(x))}(t_0(x)+ut_1(x)))\\
   &=f_0(x)\varphi^{-1}(x^{-\deg(t_0(x))}t_0(x))+f_0(x)\varphi^{-1}(x^{-\deg(t_0(x))}ut_1(x))\\
   &~~~~~~~~~~~~~~~~~~~~~~~~~~~~~~~~~~~~~~+uf_1(x)\varphi^{-1}(x^{-\deg(t_0(x))}t_0(x))\\
   &=f_0(x)\frac{x^n-\lambda}{g_1(x)} l_1(x)+f_0(x)\varphi^{-1}(x^{-\deg(t_0(x))}ut_1(x))+uf_1(x)\frac{x^n-\lambda}{g_1(x)} l_1(x)\\
   &~~~~~~~~~~~~~~~~~~~~~~~~~~~~~~~~~~~~~~(\text{using  (\ref{eq:4.16})})\\
   &=\frac{f_0(x)}{g_1(x)}(x^n-\lambda) l_1(x)+f_0(x)\varphi^{-1}(x^{-\deg(t_0(x))}ut_1(x))+ f_0(x)u
   m(x)l_1(x)\\
   &~~~~~~~~~~~~~~~~~~~~~~~~~~~~~~~~~~~~~~(\text{using (\ref{eq:frac-skewR}),  (\ref{eq:4.2-frac-xn-a-g1}) and (\ref{eq:4.5})})\\
   &=\frac{f_0(x)}{g_1(x)} l_1(x)(x^n-\lambda)+f_0(x)\left(\varphi^{-1}(x^{-\deg(t_0(x))}ut_1(x))+ um(x)l_1(x)\right)  \\
&=f_0(x)(\varphi^{-1}(x^{-\deg(t_0(x))}ut_1(x))+ um(x)l_1(x)).\\
\end{align*}
Then there exists $l_2(x)\in \mathbb{F}_{p^m}[x;\theta]$ such that, in $( \mathbb{F}_{p^m}+u\mathbb{F}_{p^m})
  [x;\Theta],$
\begin{align}
  f_0(x)(\varphi^{-1}(x^{-\deg(t_0(x))}ut_1(x))+
  um(x)l_1(x))&= ul_2(x)(x^n-\lambda)\notag \\&
                      =(x^n-\lambda)ul_2(x).\label{eq-deg2}
\end{align}
Using (\ref{eq:deg-m}),  (\ref{eq:deq-t0}) and the fact that $\deg(f_0(x))>\deg(f_1(x))$, we conclude that
\begin{align}\label{eq:4.19}\deg(m(x)l_1(x))\leq
\deg(m(x))+\deg(l_1(x))<\deg(t_0(x)) .
\end{align}
Hence, from (\ref{eq-deg2}) and (\ref{eq:4.19}),
\begin{align}\label{eq:deg-t0-2}\deg(t_0(x))=n+ \deg(l_2(x))-\deg(f_0(x)).\end{align}
The left cancellation of (\ref{eq-deg2}) by  $f_0(x)$ implies
\begin{align*}
   \varphi^{-1}(x^{-\deg(t_0(x))}ut_1(x))+
  um(x)l_1(x)= \frac{x^n-\lambda}{f_0(x)}ul_2(x) .
\end{align*}
Hence $
  \varphi^{-1}(x^{-\deg(t_0(x))}ut_1(x)) = \displaystyle\frac{x^n-\lambda}{f_0(x)}ul_2(x) -
  um(x)l_1(x),
$ i.e.,
\begin{align}\label{eq:ut1}
ut_1(x)=x^{ \deg(t_0(x))}\varphi\left(\frac{x^n-\lambda}{f_0(x)}ul_2(x) -
  um(x)l_1(x)\right).
\end{align}
Therefore,
\begin{align*}
t_0(x)+ut_1(x)&=x^{ \deg(t_0(x))}\varphi(
l_1(x))\varphi\left(\frac{x^n-\lambda}{g_1(x)}\right) \\
&~~~ +x^{ \deg(t_0(x))}\varphi\left(\frac{x^n-\lambda}{f_0(x)}ul_2(x)
-um(x)l_1(x)\right)~(\text{using (\ref{eq:t0}) and (\ref{eq:ut1})})\\
   &=x^{ \deg(t_0(x))}\varphi(
l_1(x))\varphi\left(\frac{x^n-\lambda}{g_1(x)}\right) -x^{ \deg(t_0(x))}\varphi(l_1(x))\varphi(
  um(x)) \\
&~~~+x^{
\deg(t_0(x))}\varphi\left(\frac{x^n-\lambda}{f_0(x)}ul_2(x)\right)\\
  &=x^{ n+\deg(l_1(x))-\deg(g_1(x))}\varphi(
l_1(x)) \varphi\left(\frac{x^n-\lambda}{g_1(x)}  -
  um(x) \right) \\
&~~~+x^{n+ \deg(l_2(x))-\deg(f_0(x))}\varphi(l_2(x))\varphi(\frac{x^n-\lambda}{f_0(x)}u)~~~ (\text{using (\ref{eq:deq-t0}) and (\ref{eq:deg-t0-2})})\\
 &=x^{\deg(l_1(x))} {\overrightarrow{\varphi(
l_1(x))}}^{x^{ n -\deg(g_1(x))}}  x^{ n -\deg(g_1(x))}\varphi\left( \frac{x^n-\lambda}{g_1(x)} -
  um(x) \right) \\
&~~~+x^{\deg(l_2(x))} {\overrightarrow{\varphi(l_2(x))}}^{x^{n -\deg(f_0(x))
  }} x^{n
-\deg(f_0(x))
  }\varphi(\frac{x^n-\lambda}{f_0(x)}u)\in D.
\end{align*}
As desired, $C^\perp\subseteq D$. \hfill $\square$


\subsection{Hermitian Dual  Codes of Skew Cyclic and Skew Negacyclic Codes over
$\mathbb{F}_{p^m}+u\mathbb{F}_{p^m}$}

We assume that the order of $\Theta$ is $2$ and determine the structure of the Hermitian dual codes of skew
cyclic and skew negacyclic codes in terms of their  unique representative generators,
  the ring anti-monomorphism
$\varphi$ defined in Proposition~\ref{prop2.7}  and the ring automorphism $\phi$  defined in (\ref{eq-lem-phi}).
Using Lemma \ref{lem2.10} and arguments similar to those in the previous subsection, the next theorem follows.

\begin{theorem}\label{thm4.6}
Let  $\lambda\in \{1,-1\}$ and let $\Theta$ be an automorphism of order   $2$. Then the Hermitian dual code of a
left ideal  in $(\mathbb{F}_{p^m}+u\mathbb{F}_{p^m})[x;\Theta]/\langle x^n-\lambda\rangle $ is again a left
ideal in $(\mathbb{F}_{p^m}+u\mathbb{F}_{p^m})[x;\Theta]/\langle x^n-\lambda\rangle $ determined as   follows:
\begin{enumerate}[\textnormal{LI}-$1^{\perp_H}.$]
  \item If  $C=\langle  g_0(x)+ug_1(x)\rangle $,  then $C^{\perp_H}=\langle  \phi(x^{n-\deg(g_0(x))}\varphi\left(
  \displaystyle\frac{x^n-\lambda}{ g_0(x)+ug_1(x)}\right))\rangle
  $.
   \item If $C=\langle  ug_1(x) \rangle$,  then   $C^{\perp_H}=\langle u, \phi (x^{n-\deg( g_1(x))} \varphi\left(
     \displaystyle\frac{x^n-\lambda}{ g_ 1(x)} \right)) \rangle
  .$ %
  \item If $C=\langle  ug_1(x), f_0(x)+uf_1(x) \rangle$, then there exists $m(x)\in \mathbb{F}_{p^m}[x;\theta] $ such that $m(x)g_1(x)= {\overleftarrow{\left(\displaystyle\frac{x^n-\lambda}{ f_0(x) }\right)}}^u f_1(x)$ and  $$C^{\perp_H}=\langle \phi ( x^{n-\deg( f_0(x))}\varphi\left(
   \displaystyle\frac{x^n-\lambda}{ f_0(x)}u\right)), \phi(x^{n-\deg( g_1(x))} \varphi\left(
     \displaystyle\frac{x^n-\lambda}{ g_ 1(x)}-um(x)\right))\rangle
  .$$

\end{enumerate}
 \end{theorem}

\begin{example}\label{ex-last}
Table \ref{tab:dual} shows the Euclidean and Hermitian dual codes of the left  ideals in
$(\mathbb{F}_{3}+u\mathbb{F}_{3})[x;\Theta_{{\rm id},2}]/\langle x^2-1\rangle$ classified in Example
\ref{ex4.2}.
  The dual codes are obtained via Theorems
\ref{thm4.5} and \ref{thm4.6} and   rewritten to satisfy  the representation in Proposition \ref{prop4.1}. The
subscripts $1$, $2$ and $3$ indicate types LI-$1$, LI-$2$ and LI-$3$, respectively.
\begin{table}[hbt]
  \caption{The left ideals in $(\mathbb{F}_3+u\mathbb{F}_3)[x;\Theta_{{\rm id},2}]/\langle x^2-1\rangle$ and their Euclidean and Hermitian dual codes}\label{tab:dual}
  \centering
  \begin{tabular}{lll}
    \hline
    \hline
    $C$ ~~~~~~~~~~~~~~~~~~~  & $C^\perp$ ~~~~~~~~~~~~~~~~~~~ & $C^{\perp_H}$ ~~~~~~~~~~~~  \\
    \hline
    $\langle 0\rangle_1$        &   $\langle 1\rangle_1$    &   $\langle 1\rangle_1$  \\
    $\langle u(x+1)\rangle_2$   &$\langle u, x+2  \rangle_3$&$\langle u, x+2  \rangle_3$\\
    $\langle u(x+2)\rangle_2$   &$\langle u, x+1  \rangle_3$&$\langle u, x+1  \rangle_3$\\
    $\langle u \rangle_2$       &$\langle u \rangle_2$&$\langle u \rangle_2$\\
    $\langle x+1+2u \rangle_1$  &$\langle x+2+2u  \rangle_1$&$\langle x+2+u  \rangle_1$\\
    $\langle x+1+u \rangle_1$   &$\langle x+2+u  \rangle_1$&$\langle x+2+2u  \rangle_1$\\
    $\langle x+1  \rangle_1$    &$\langle x+2  \rangle_1$&$\langle x+2  \rangle_1$\\
    $\langle x+2  \rangle_1$    & $\langle x+1  \rangle_1$ & $\langle x+1  \rangle_1$ \\
    $\langle x+2+u  \rangle_1$  &$\langle x+1+u \rangle_1$&$\langle x+1+2u \rangle_1$\\
    $\langle x+2+2u  \rangle_1$ &$\langle x+1+2u \rangle_1$&$\langle x+1+u \rangle_1$\\
    $\langle u, x+1  \rangle_3$ &$\langle u(x+2)\rangle_2$ &$\langle u(x+2)\rangle_2$ \\
    $\langle u, x+2  \rangle_3$ &$\langle u(x+1)\rangle_2$&$\langle u(x+1)\rangle_2$\\
    $\langle 1\rangle_1$        &$\langle 0\rangle_1$ &$\langle 0\rangle_1$ \\
    \hline
  \end{tabular}
\end{table}
\end{example}

\vspace{-1cm}

\section{Conclusion}
The concept of coding with skew polynomial rings over finite fields \cite{BoGe2007} and \cite{BoUl2009} and over
Galois rings \cite{BoSoUm2008}   is extended to the case over finite chain rings. Given an automorphism $\Theta$
of a finite chain ring $\mathcal{R}$ and a unit $\lambda$ in $\mathcal{R}$, $\Theta$-$\lambda$-constacyclic
codes are introduced. Under the assumptions that $\lambda$ is a unit fixed by $\Theta$ and the length $n$ of
codes is a multiple of the order of $\Theta$, $\Theta$-$\lambda$-constacyclic codes can be viewed as left ideals
in the quotient ring $\mathcal{R}[x;\Theta]/\langle x^n-\lambda\rangle$. In particular, when the code is
generated by a monic right divisor $g(x)$ of $x^n-\lambda$, its properties are exhibited. When $\lambda^2=1$,
the generators of its Euclidean and Hermitian dual codes are given in terms of
$h(x):=\displaystyle\frac{x^n-\lambda}{g(x)}$. Moreover, necessary and sufficient conditions for a
$\Theta$-$\lambda$-constacyclic code to be Euclidean and Hermitian self-dual are provided.

A typical example of a finite chain ring is $\mathbb{F}_{p^m}+u\mathbb{F}_{p^m}+\dots+ u^{e-1}\mathbb{F}_{p^m}$.
In the case $e=2$, a complete classification of $\Theta$-$\lambda$-constacyclic codes over the ring
$\mathbb{F}_{p^m}+u\mathbb{F}_{p^m}$ is given. For the special case when $\lambda=\pm 1$, the classification
provides generators of the Euclidean and Hermitian dual codes of skew cyclic and skew negacyclic codes  based on
generators of the original codes. Moreover, an illustration of all skew cyclic codes of length~$2$ over
$\mathbb{F}_{3}+u\mathbb{F}_{3}$ and their Euclidean and Hermitian dual codes is also provided.

For   further work, using the idea in \cite{BoUl2009-module}, constructions and classification of skew
constacyclic codes over finite chain rings could be considered   as modules over the skew polynomial ring
$\mathcal{R}[x;\Theta]$. This may lead to classification of codes of arbitrary lengths and constructions of more
codes with good parameters.


\medskip
Received xxxx 20xx; revised xxxx 20xx.
\medskip


\begin{thebibliography}{99}
\bibitem{Al-1991} Y.  Alkhamees,  \textit{The determination of the group of
automorphisms of a finite chain ring of characteristic $p$}, {The Quarterly Journal of Mathematics}, \textbf{42}
(1991), 387-391.

\bibitem{Al-1990} Y. Alkhamees, \textit{The group of automorphisms of
finite chain rings},  {Arab Gulf Journal of Scientific Research}, \textbf{8} (1990),  17-28.

\bibitem{AmNe2008} (MR2459836) M. C. V. Amarra and  F. R. Nemenzo,    \textit{On $(1-u)$-cyclic codes over $\mathbb{F}_{p^k}+u\mathbb{F}_{p^k}$},
        {Applied Mathematics Letters},  \textbf{21} (2008), 1129-1133.

\bibitem{Ba1997} (MR1439633) C. Bachoc, \textit{Application of coding theory to the construction of modular lattices},
Journal of Combinatorial Theory Series A, \textbf{78} (1997), 92--119.

\bibitem{BiFl2002} (MR1919698) G. Bini and   F. Flamini,  \lq\lq Finite Commutative Rings and
Their Applications," Kluwer Academic Publishers, Massachusetts, 2002.

\bibitem{BoUd1999} (MR1686262) A. Bonnecaze and P. Udaya, \textit{Cyclic codes and self-dual codes over
$\mathbb{F}_2+u\mathbb{F}_2$}, IEEE Transactions on Information Theory, \textbf{45}  (1999), 1250--1255.


\bibitem{BoGe2007} (MR2322946) D. Boucher,  W. Geiselmann and F.  Ulmer,   \textit{Skew-cyclic codes}, {Applicable Algebra in  Engineering,  Communication and Computing}, \textbf{18} (2007), 379-389.

\bibitem{BoSoUm2008} (MR2429458) D. Boucher, P. Sol\'e and F. Ulmer, \textit{Skew constacyclic codes over Galois rings},
{Advances in Mathematics of Communications}, \textbf{2} (2008),
 273-292.

\bibitem{BoUl2009-module} D. Boucher and F. Ulmer,   \textit{Codes as modules over skew polynomial rings},
{Lecture Notes in Computer Science},  \textbf{5921} (2009), 38--55.

\bibitem{BoUl2009} (MR2553570) D. Boucher and  F. Ulmer,   \textit{Coding with skew polynomial rings},
{Journal of Symbolic Computation},  \textbf{44} (2009), 1644--1656.






\bibitem{ClDr1973} (MR0332875) W. E. Clark and   D. A. Drake,   \textit{Finite chain rings},
{Abhandlungen aus dem Mathematischen Seminar der Universit\"at Hamburg}, \textbf{39} (1973), 147-153.


\bibitem{ClLi1973} (MR0337910) W. E. Clark and J. J. Liang,  \textit{Enumeration of finite
commutative chain rings},  {Journal of Algebra}, \textbf{27} (1973), 445-453.

\bibitem{HQDSRLP2004}(MR2096841) H. Q. Dinh and S. R. L\'{o}pez-Permouth,  \textit{Cyclic and negacyclic codes over finite chain
rings}, {IEEE Transactions on  Information Theory}, \textbf{50} (2004), 1728--1744.


\bibitem{Di2005}(MR2243156) H. Q. Dinh,   \textit{Negacyclic codes of length $2^s$ over Galois rings},
         {IEEE Transactions on  Information Theory},  {\bf 51} (2005), 4252--4262.

\bibitem{Di2009} (MR2582760) H. Q. Dinh,   \textit{Constacyclic codes of length $2^s$ over Galois
        extension rings of $\mathbb{F}_2+u\mathbb{F}_2$}, {IEEE Transactions on Information Theory},
          \textbf{55}  (2009), 1730--1740.

\bibitem{Di2010} H. Q. Dinh, \textit{Constacyclic codes of length $p^s$ over $\mathbb{F}_{p^m} + u\mathbb{F}_{p^m}
$}, Journal of Algebra, \textbf{324} (2010), 940--950.


\bibitem{HaKu1994}(MR1294046)  A. R. Hammons,  P. V. Kumar,  A. R.  Calderbank, N. J. A. Sloane and  P. Sol\'{e},     \textit{The
$\mathbb{Z}_4$-linearity of Kerdock, Preparata, Goethals and related codes}, {IEEE Transactions on Information
Theory}, \textbf{40} (1994),  301--319.


\bibitem{Mc1974} (MR0354768) B. R. McDonald,   \lq\lq Finite Rings with Identity,"   Marcel Dekker, New York 1974.


\bibitem{QiZa2006}(MR2232261) J. F. Qian,   L. N. Zhang and  S. X.    Zhu,   \textit{$(1+u)$-cyclic and cyclic
        codes over the ring $\mathbb{F}_{2}+u\mathbb{F}_{2}$},
         {Applied Mathematics Letters}, \textbf{19} (2006), 820--823.




\bibitem{NoSa2000}(MR1770877) G. H. Norton  A. S\u{a}l\u{a}gean,   \textit{On the structure of linear and cyclic codes over
a finite chain ring},   {Applicable Algebra in  Engineering, Communication and Computing}, \textbf{10} (2000),
489--506.

\bibitem{Ri1970}(MR0396640) P. Ribenboim,  \textit{Sur la localisation des anneaux non commutatifs}
(French), {S\'eminaire Dubreil. Alg\`{e}bre et th\'eorie des nombres},  \textbf{24} (1970), 1970/71.



\bibitem{SoEs2009-arbi} (MR2537491) R. Sobhani, and  M. Esmaeili,    \textit{Cyclic and negacyclic codes over the Galois ring $\textnormal{GR}(p^2,m)$},
         {Discrete Applied Mathematics},  {\bf 157} (2009),   2892--2903.

\bibitem{UdBo1999}(MR1720674) P. Udaya and A. Bonnecaze, \textit{Decoding of cyclic codes over $\mathbb{F}_2 +u\mathbb{F}_2$},
 IEEE Transactions on Information Theory,
\textbf{45} (1999), 2148--2157.

\bibitem{Udsi1998}(MR1665807) P. Udaya and M. U. Siddiqi, \textit{Optimal large linear complexity frequency hopping patterns derived from
polynomial residue class rings}, IEEE Transactions on Information Theory, \textbf{44} (1998), 1492--1503.


\bibitem{Wa2003}(MR2008834)  Z.-X. Wan, \lq\lq Lectures on Finite Fields and Galois Rings," World
Scientific, New Jersey,  2003.


\end{thebibliography}
\end{document}